\documentclass[journal,onecolumn,twoside]{IEEEtran}
\normalsize
\ifCLASSINFOpdf
\else
\fi
\usepackage{setspace}
\usepackage{color}
\usepackage{cite}
\usepackage{amsmath}
\usepackage{amsfonts}
\usepackage{amssymb}
\usepackage{amsthm}
\usepackage{tikz}
\usepackage{cases}
\usepackage{bm}
\usepackage{graphicx}
    \graphicspath{{../}}
    \DeclareGraphicsExtensions{.pdf}
\usepackage[caption=true,font=footnotesize]{subfig}
\usepackage{multirow}
\usepackage{makecell}
\usepackage{mathdots}
\usepackage{booktabs}
\usepackage{url}
\usepackage{bm}
\usepackage{xtab}
\usepackage{pifont}
\usepackage{array}
\usepackage{algorithm}
\usepackage{algorithmic}
\usepackage{enumerate}
\usepackage{extarrows}
\usepackage{tikz}
\usepackage{amsfonts}
\usepackage{color}
\usepackage{xcolor}
\usetikzlibrary{arrows}
\usepackage{caption}

\usepackage{lineno} 

\usepackage{mathrsfs}
\usepackage{tikz-cd}
\usepackage{mathtools}
\usepackage{makecell}
\usepackage{booktabs}
\usepackage{array}
\usepackage{enumitem}

\allowdisplaybreaks[4]
\usepackage[colorlinks,
            linkcolor=blue,
            anchorcolor=blue,
            citecolor=blue,
            urlcolor=black]{hyperref}
\theoremstyle{plain}

\newtheorem{theorem}{Theorem}[section]
\newtheorem{lemma}{Lemma}[section]
\newtheorem{proposition}{Proposition} [section]
\newtheorem{definition}{Definition} [section]
\newtheorem{corollary}{Corollary} [section]
\theoremstyle{definition}

\newtheorem{remark}{Remark}[section]

\DeclareMathOperator{\Supp}{supp}

\newcommand{\Span}{\mathrm{span}}

\newcommand{\Bcal}{\mathcal{B}}
\newcommand{\Ccal}{\mathcal{C}}

\newcommand{\Ecal}{\mathcal{E}}

\newcommand{\Hcal}{\mathcal{H}}

\newcommand{\Lcal}{\mathcal{L}}

\newcommand{\Ocal}{\mathcal{O}}
\newcommand{\Pcal}{\mathcal{P}}

\newcommand{\Rcal}{\mathcal{R}}

\newcommand{\Vcal}{\mathcal{V}}

\newcommand{\Xcal}{\mathcal{X}}
\newcommand{\Ycal}{\mathcal{Y}}

\newcommand{\Pbb}{\mathbb{P}}

\newcommand{\Zbb}{\mathbb{Z}}

\newcommand{\pfrac}[2]{\frac{\partial #1}{\partial #2}}

\newcommand{\F}{\mathbb{F}}

\newcommand{\Div}{\mathrm{Div}}

\newcommand{\Char}{\mathrm{char}}

\newcommand{\Emph}{}

\newcommand{\undetermined}{11}

\newcommand{\enco}{\mathrm{enc}}
\newcommand{\info}{\mathrm{info}}
\newcommand{\query}{\mathrm{query}}
\newcommand{\noise}{\mathrm{noise}}
\newcommand{\priv}{\mathrm{priv}}
\newcommand{\secu}{\mathrm{sec}}
\newcommand{\full}{\mathrm{full}} 
\newcommand{\vp}{\varphi}

\DeclarePairedDelimiter\efloor{\lfloor}{\rfloor}

\DeclarePairedDelimiter\eparentheses{(}{)}
\DeclarePairedDelimiter\ebrackets{[}{]}
\DeclarePairedDelimiter\ebraces{\{}{\}}

\newcommand{\floor}{\efloor*}

\newcommand{\parentheses}{\eparentheses*}
\newcommand{\brackets}{\ebrackets*}
\newcommand{\braces}{\ebraces*}

\newcommand{\new}{\mathrm{new}}

\makeatletter

\newcommand{\Rmnum}[1]{\expandafter\@slowromancap\romannumeral #1@}
\makeatother
 
\newcommand{\Rnum}[1]{\uppercase\expandafter{\romannumeral #1\relax}}
\newcommand{\rnum}[1]{\expandafter{\romannumeral #1\relax}}

\newcommand{\Ev}{\mathrm{ev}}

\begin{document}

\title{New $X$-Secure $T$-Private Information Retrieval Schemes via Rational Curves and Hermitian Curves} 
\author{Yuan~Gao,~\IEEEmembership{} 
        Weijun~Fang, ~\IEEEmembership{}
        Jingke~Xu, ~\IEEEmembership{}  
        Jiejing~Wen ~\IEEEmembership{}
        
\IEEEcompsocitemizethanks{\IEEEcompsocthanksitem Yuan Gao, Weijun Fang, and Jiejing Wen are with the State Key Laboratory of Cryptography and Digital Economy Security, Shandong University, Qingdao, 266237, China, the Key Laboratory of Cryptologic Technology and Information Security, Ministry of Education, Shandong University, Qingdao, 266237, China and the School of Cyber Science and Technology, Shandong University, Qingdao, 266237, China. Jingke Xu is with the School of Information Science and Engineering, Shandong Agricultural University, Tai'an, 271018, China (emails: gaoyuan862023@163.com, fwj@sdu.edu.cn,
xujingke@sdau.edu.cn, 
jjwen@sdu.edu.cn).
}
\thanks{ This research is supported in part by the National Key Research and Development Program of China under Grant Nos. 2022YFA1004900 and 2021YFA1001000, the National Natural Science Foundation of China under Grant Nos. 62571301 and 12201357, the Natural Science Foundation of Shandong Province under
Grant Nos. ZR2025MS59 and ZR2022QF001. 
{\it (Corresponding Author: Weijun Fang and Jingke Xu)}.}
\thanks{Manuscript submitted }}

\maketitle
\begin{abstract}  
$X$-secure and $T$-private information retrieval (XSTPIR) is a variant of private information retrieval where data security is guaranteed against collusion among up to $X$ servers and the user's retrieval privacy is guaranteed against collusion among up to $T$ servers. Recently, researchers have constructed XSTPIR schemes through the theory of algebraic geometry codes and algebraic curves, with the aim of obtaining XSTPIR schemes that have higher maximum PIR rates for fixed field size and $X,T$ (the number of servers $N$ is not restricted). The mainstream approach is to employ curves of higher genus that have more rational points, evolving from rational curves to elliptic curves to hyperelliptic curves and, most recently, to Hermitian curves.

In this paper, we propose a different perspective: with the shared goal of constructing XSTPIR schemes with higher maximum PIR rates, we move beyond the mainstream approach of seeking curves with higher genus and more rational points. Instead, we aim to achieve this goal by enhancing the utilization efficiency of rational points on curves that have already been considered in previous work. By introducing a family of bases for the polynomial space $\Span_{\F_q}\{1,x,\dots,x^{k-1}\}$ as an alternative to the Lagrange interpolation basis, we develop two new families of XSTPIR schemes based on rational curves and Hermitian curves, respectively. 
Parameter comparisons demonstrate that our schemes achieve superior performance. Specifically, our Hermitian-curve-based XSTPIR scheme provides the largest known maximum PIR rates when the field size $q^2\geq 14^2$ and $X+T\geq 4q$. Moreover, for any field size $q^2\geq 28^2$ and $X+T\geq 4$, our two XSTPIR schemes collectively provide the largest known maximum PIR rates.
\end{abstract}


\IEEEpeerreviewmaketitle

\section{Introduction}
The problem of private information retrieval (PIR)  was first proposed by Chor \textit{et al.} \cite{CKGS95FOCS:PIR} in 1995. 
It is a canonical problem in the study of privacy issues that arise from the retrieval of information from public databases. 

In practice, it is common that each file is quite large, making the download from the servers the dominant communication cost. Motivated by this, the PIR problem was reformulated in \cite{Sun&Jafar16:CapacityPIR} from the perspective of information theory. This formulation \cite{Sun&Jafar16:CapacityPIR} assumes arbitrarily large message sizes, thereby allowing the upload cost to be neglected.
Within this framework, the efficiency of a PIR scheme is measured by the \Emph{PIR rate}, which is defined to be the ratio of the size of the requested message to the total number of downloaded bits. The PIR capacity is defined to be the supremum of the PIR rates of all feasible schemes. 
Considerable progress has been made on various variants of the PIR problem in recent years, such as
\cite{Sun&Jafar16:ColludPIR, YLW20:CapaPIRCollud, banawan2019capacity, sun2019capacitySym, wang2019PIRSym, Banawan2018multi, zhang2019general, hollanti2017private, tajeddine2019private, Holzbaur2022towards, lin2018MDSPIR, xu2018sub, xu2018building, xu2019capacityachiev, xu2022building, xu2025explicit, hollanti2019tprivate, sunrui2025MDSTPIR, jia2020x}. A systematic review of the literature is available in \cite{UAGJ22:SuPIR}.

$X$-secure and $T$-private information retrieval (XSTPIR) is a variant of private information retrieval where data security is guaranteed against collusion among up to $X$ servers and the user's privacy is guaranteed against collusion among
up to $T$ servers (see Section~\ref{sec:2.2} for details). In \cite{jia2019cross}, the asymptotic (large file number $K$) capacity of XSTPIR schemes was determined, where the idea of Cross-Subspace Alignment (CSA) was applied to design XSTPIR schemes that attain the capacity.  
Later in \cite{makkonen2024algebraic}, the original construction of the XSTPIR schemes in \cite{jia2019cross}, which employed Reed–Solomon codes, was reinterpreted by a new general framework, which employed algebraic-geometry codes (AG codes) from algebraic curves.
Based on this general framework, they \cite{makkonen2024algebraic} recovered the construction of XSTPIR schemes in \cite{jia2019cross} using AG codes from rational curves (genus 0), and further proposed new XSTPIR schemes using AG codes from elliptic curves (genus 1).
As pointed out in \cite{makkonen2024algebraic}, when the field size $q$, the secure parameter $X$, the privacy parameter $T$ are all fixed (and the number of servers $N$ is not restricted), the XSTPIR schemes based on elliptic curves may have higher maximum PIR rates\footnote{In this paper, whenever we discuss the maximum PIR rate of a family of XSTPIR schemes, we refer to the maximum PIR rate achievable by this family of XSTPIR schemes for fixed field size $q$ (or $q^2$), and fixed $X,T$, where the number of servers $N$ is unrestricted.} than the XSTPIR schemes based on rational curves (see \cite[Fig. 1]{makkonen2024algebraic}, or see \cite[Proposition 4.2]{ghiandoni2025agcodes}).
This improvement of maximum PIR rates is owing to the fact that increasing the genus brings curves with more rational points. 
Extending the results of \cite{makkonen2024algebraic}, Makkonen \textit{et al.} \cite{makkonen2024secretsharingsecureprivate} constructed XSTPIR schemes based on hyperelliptic curves of genus $g\geq 1$, further improving the maximum PIR rates in some cases. 
Very recently in \cite{ghiandoni2025agcodes}, Ghiandoni \textit{et al.} proposed a family of new XSTPIR schemes based on Hermitian curves. Their schemes have significantly higher maximum PIR rates than previous schemes in some cases, especially when $X+T$ is large.

\subsection{Our Motivation and Contributions}  
This work investigates new constructions of XSTPIR schemes that achieve higher maximum PIR rates than existing schemes for fixed field size, security parameter $X$, and privacy parameter $T$ (with no restriction on the number of servers $N$). To this end, we briefly review the evolution of existing constructions.
To obtain XSTPIR schemes with higher maximum PIR rates (for fixed field size and $X,T$), the prevailing approach is to employ algebraic curves of higher genus that possess more rational points. 
The underlying algebraic curves for XSTPIR schemes have moved based on rational curves \cite{jia2019cross, makkonen2024algebraic} to elliptic curves \cite{makkonen2024algebraic}, to hyperelliptic curves \cite{makkonen2024secretsharingsecureprivate}, and further to Hermitian curves \cite{ghiandoni2025agcodes}.
 
Instead of pursuing the conventional path of seeking the help of curves with more and more rational points, we re-examine established constructions carefully and find that the rational points of their underlying curves can be leveraged more effectively.   
By introducing a family of new bases of the space $\Span_{\F_q}\{1,x,\dots,x^{k-1}\}$ (see Section~\ref{sec:3.1}), we develop new XSTPIR schemes that utilize rational points more efficiently, thus increasing the maximum PIR rates. Concretely, we propose two families of XSTPIR schemes using rational curves (Theorem~\ref{thm:our_XSTPIR_Rational}) and hermitian curves (Theorem~\ref{thm:our_XSTPIR_Hermitian}), respectively. 
It is shown in Section~\ref{sec:4} that when the field size $q^2\geq 14^2$ and $X+T\geq 4q$, our XSTPIR schemes based on Hermitian curves provide the largest known maximum PIR rates; and when the field size $q^2\geq 28^2$ and $X+T\geq 4$, our two constructions collectively outperform all existing XSTPIR schemes, providing the largest maximum PIR rates (See Corollary~\ref{cor:comparison_our_hermitian}, Corollary~\ref{cor:comparison_our_rational_and_hermitian}, and Figure~\ref{fig:comparison}).   

It is worth noting that the superior performance of our new scheme's maximum PIR rates is not achieved by seeking more advanced curves, but rather by better utilizing curves already employed by prior works. We believe that our approach offers valuable insights for future research, even when applied to novel curves that have yet to be explored for XSTPIR scheme constructions.
\subsection{Paper Organization}
The remainder of this paper is organized as follows. In Section \ref{sec:2}, we introduce the basic notions 
about algebraic function fields, algebraic geometry codes, and some algebraic curves that will be employed in this paper. Also, the concept of $X$-secure and $T$-private information retrieval is recalled. The general framework for producing XSTPIR schemes proposed by Makkonen \textit{et al.} \cite{makkonen2024algebraic} is introduced, along with all known explicit XSTPIR schemes based on it. Moreover, some results on the maximum PIR rates of these schemes are also included.
In Section \ref{sec:3}, based on the above-mentioned general framework, we propose two families of XSTPIR schemes by using rational curves and Hermitian curves.
In Section~\ref{sec:4}, we determined the maximum PIR rates of our proposed two XSTPIR schemes for fixed field size, security parameter $X$, and privacy parameter $T$. And then we compare the maximum PIR rates of our proposed two XSTPIR schemes with all existing XSTPIR schemes. In Section~\ref{sec:5}, we conclude this paper.
\section{Preliminaries}
\label{sec:2}
 Before proceeding, we fix some notations.
\begin{itemize}
    \item For two real numbers $a\leq b$, we define $[a,b]:=\{n\in \Zbb:\; a\leq n\leq b\}$. For positive integer $a$, we define $[a]:=\{1,\dots,a\}.$
    \item We use $\lfloor \cdot\rfloor$ and $\lceil \cdot\rceil$ to denote the floor and the ceiling function, respectively.
\end{itemize}
\subsection{Background on Algebraic Curves, Algebraic Function Fields, and Algebraic-Geometry Codes}
In this subsection, we present some preliminaries on algebraic curves, algebraic function fields, and algebraic geometry codes. We primarily adopt the notation and conventions of the theory of algebraic function fields as introduced in \cite{stichtenoth2009algebraic}, and see \cite{niederreiter2001rational,niederreiter2009algebraic} for some omitted details, such as the correspondence between places and $\F_q$-closed points. 

Let $\Vcal$ \footnote{Here we use $\Vcal$ rather than the more convenient symbol $\Xcal$, since the symbol $\Xcal$ will be later used to denote the rational curve specially.} be a projective, absolutely irreducible, non-singular algebraic curve defined over the finite field $\F_q$ of genus $g(\Vcal)$. 
Let $F=\F_q(\Vcal)$ be the function field of the curve $\Vcal/\F_q$, and let $\Pbb_F$ denote the set of all places of $F$. There is a one-to-one correspondence between $\Pbb_F$ and the set of all $\F_q$-closed points of $\Vcal$. In particular, the set of all rational places $\Pbb_{F}^1$ is in one-to-one correspondence with the set of all $\F_q$-rational points \footnote{$\F_q$-rational points are those $\F_q$-closed points with cardinality $1$.} of $\Vcal$. 

The free abelian group generated by $\Pbb_F$ is called the divisor group of $F/\F_q$ and is denoted by $\Div(F)$. Every element $D$ of $\Div(F)$ can be written as the form $D=\sum_{P\in \Pbb_F} n_P P$, where $n_P\neq 0$ for finitely many $P\in \Pbb_F$. The support set of $D$ is defined by $\Supp(D):=\{P\in \Pbb_F:n_P\neq 0\}$; and the degree of $D$ is defined by $\deg(D):=\sum_{P\in \Pbb_F} \deg(P)$. 
For a function $f \in F$, we use $(f)_0$ and $(f)_{\infty}$ to denote its zero divisor and pole divisor, respectively. Its principal divisor is defined by $$(f):=(f)_0-(f)_{\infty}.$$
The Riemann-Roch space associated to a divisor $D\in \Div(F)$ is defined by
$$\mathcal{L}(D):=\{f\in F:\;\ (f)+D \geq 0\}\cup \{0\},$$
Its dimension as a vector space over $\F_q$ is denoted by $\ell(D)$. Let us recall some useful properties of Riemann-Roch spaces that will be applied in this paper: 
\begin{itemize}
    \item $f\cdot\Lcal(D)=\Lcal(D-(f))$, where $f\in F$ and $f\cdot\Lcal(D)$ denotes $\{f\cdot u:\;u\in \Lcal(D)\}$.
    \item $\Lcal(D) \cdot\Lcal(D') \subseteq \Lcal(D+D')$, where $\Lcal(D) \cdot\Lcal(D'):=\text{span}_{\F_q}\{f\cdot g \,:\, f\in \Lcal(D), g\in \Lcal(D')\}$.
    \item If $D \le D'$, then $\Lcal(D) \subseteq  \Lcal(D')$.
\end{itemize}

Let $\Pcal=\{P_1,\cdots,P_n\}\subseteq \Pbb_F^1$ be a set of pairwise distinct rational places. Take another divisor $G\in \Div(F)$ such that $\Supp(G)\cap \Pcal=\varnothing$. The AG code $\mathcal{C}(\Pcal,G)$ associated with $\Pcal$ and $G$ is a linear subspace of $\F_q^n$ which is defined as the image of the evaluation map $\Ev:\;\mathcal{L}(G) \to \F_q^n$ given by $\Ev(f) = (f(P_1),f(P_2) ,\ldots,f(P_n))$. The code $\mathcal{C}(\Pcal,G)$ has length $n$. 
\begin{lemma}[Part of {\cite[Theorem 2.2.2, Corollary 2.2.3]{stichtenoth2009algebraic}}]\label{lem:AGcode_lem1}
    Let $\Pcal,G, \Ev, \Ccal(\Pcal,G)$ be as defined above. 
         If $2g(\Vcal)-2<\deg(G)<n$, then the map $\Ev$ is injective and the dimension of $\mathcal{C}(\Pcal,G)$ is equal to $\ell(G)=\deg(G)+1-g(\Vcal)$. 
\end{lemma}

The Euclidean dual code of a linear code $\Ccal$ (of length $n$) is defined as $\Ccal^{\perp}:=\{\bm y\in \F_q^n:\; (\bm y,\bm c)=0 \;\text{ for all } \bm c\in \Ccal\}$, where $(\bm y,\bm c)$ denotes the Euclidean inner product.
The dual code of an algebraic geometry code is also an AG code (see \cite[Theorem 2.2.8, Proposition 2.2.10]{stichtenoth2009algebraic}). The following lemma is extracted from \cite{stichtenoth2009algebraic}.
\begin{lemma}[Corollary from {\cite[Theorems 2.2.7 and 2.2.8]{stichtenoth2009algebraic}}]\label{lem:AGcode_lem2}
    Let $\Pcal,G$ and the linear code $\Ccal(\Pcal,G)$ be as defined above. If $2g(\Vcal)-2<\deg(G)\leq n-2$, then $\Ccal(\Pcal,G)^{\perp}$ has dimension at least $1$ and minimum distance at least $\deg(G)-(2g(\Vcal)-2)$.
\end{lemma}
In the following, we briefly introduce three families of algebraic curves that will be employed in the rest of the paper: rational curves, hyperelliptic curves, and Hermitian curves.

A rational curve $\Xcal=\Pbb^1$ over $\F_q$ is the standard smooth projective genus-0 algebraic curve; its function field is just the rational function field $\F_q(x)$. In this paper, we always use $\Xcal$ to denote a rational curve, for clarity of comparison.
 
A hyperelliptic curve $\mathcal{Y}$ of genus $g \geq 1$ over $\mathbb{F}_q$ is an algebraic curve over $\F_q$ defined by 
\[
y^2 + h(x)y = f(x),
\]
where $f$ is monic of $\deg(f) = 2g + 1$\footnote{Indeed, the degree of $f(x)$ can be of degree $2g+2$, which is omitted in this paper.} and $\deg(h)\leq g$, for some $g\geq 1$.
\begin{itemize}
    \item When $g=1$, it is already a non-singular curve, with a single point at infinity $P_{\infty}=[0,1,0]$. In this case, we also call the hyperelliptic curve $\Ycal$ an elliptic curve.
    \item When $g\geq 2$, (before normalization) it is singular only at the point $[0,1,0]$. After normalization, it also has a single point at infinity. We also denote it by $P_{\infty}$.
\end{itemize} 

The Hermitian curve $\mathcal{H}_q$ is defined over $\F_{q^2}$ by
\begin{equation}\label{eq:Hermitian}
x^{q+1}=y^q+y.
\end{equation}
It is a non-singular curve, whose genus is $g(\mathcal{H}_q)=q(q-1)/2$. It has totally $q^3+1$ $\F_{q^2}$-rational points,
including a unique point at infinity $P_{\infty}=[0:1:0]$. Thus, $\mathcal{H}_q$ is an $\F_{q^2}$-maximum curve.  

\subsection{Background and Known Constructions of $X$-Secure $T$-Private Information Retrieval Schemes}
\label{sec:2.2}
In this section, we briefly introduce the $X$-secure $T$-private information retrieval problem, which is the central topic of this paper. For an information-theoretic formulation of the $X$-secure and $T$-private information retrieval problem, see~\cite[Section~\Rmnum{2}]{jia2019cross}.

The XSTPIR problem can be briefly described as follows. 
Consider data consisting of $K$ independent messages $s_1, \dots, s_K \in \mathbb{F}_q^L$, which are stored at $N$ distributed servers. 
\begin{definition}[$X$-security]
    We say that the (storage) scheme is $X$-secure ($0\leq X< N$) if any $X$ colluding servers of all these $N$ servers learn nothing about the messages $s_1, \dots, s_K$.
\end{definition} 
Continuing with the aforementioned settings, we introduce the problem of private information retrieval. In private information retrieval, the user privately generates a desired message index $\theta\in [K]$. The user wants to retrieve $s_{\theta}$ privately.
To do so, the user generates $N$ queries $Q^{[\theta]}_1, \dots, Q^{[\theta]}_N$ and sends $Q^{[\theta]}_i$ to the $i$-th server for each $i\in [N]$. After receiving query $Q^{[\theta]}_i$, the $i$-th server returns an answer $A_i^{[\theta]}$ to the user. 
The user must be able to precisely recover the desired message $s_\theta$ from all the answers $A_1^{[\theta]}, \ldots, A_N^{[\theta]}$. 
\begin{definition}[$T$-privacy]
    We say that the (private) information retrieval scheme is $T$-private ($1\leq T\leq  N$) if any $T$ colluding servers of all these $N$ servers learn nothing about the index $\theta$ of the desired file.
\end{definition}
\begin{definition}[$X$-secure and $T$-private information retrieval scheme, and the PIR rate]
      If a storage scheme is $X$-secure and is equipped with a $T$-private information retrieval scheme, we collectively refer to them as an $X$-secure and $T$-private information retrieval scheme. 
      Equivalently, we may also say that a PIR scheme is $X$-secure and $T$-private.
      
      The PIR rate of an XSTPIR scheme is defined to be the rate of its underlying PIR scheme, which is the ratio of the file size to the total download cost. For brevity, we shall use 'rate' to refer to the PIR rate in the subsequent sections. 
\end{definition}

In \cite{makkonen2024algebraic}, the following general framework for producing XSTPIR schemes is proposed.
\begin{lemma}[{\cite[Theorem 2.1]{makkonen2024algebraic}}]\label{lem:Jia_Makk_framework}
    Let $A$ be an $\F_q$-algebra, and let $V^{\enco}_{\ell},V^{\secu}_{\ell}, V^{\query}_{\ell},V^{\priv}_{\ell}$ be finite-dimensional $\F_q$-subspaces of $A$ for any $\ell\in[L]$. Moreover, let $\varphi:A \to \F_q^N$ be an $\F_q$-algebra homomorphism such that \footnote{According to the detailed derivation in \cite[Section \Rmnum{2}]{makkonen2024algebraic}, it appears that the condition ``$\dim(V^{\enco}_{\ell})=1$ for all $\ell\in{L}$'' is also required. However, this is not a concern as all existing constructions based on this lemma, as well as our own, satisfy this condition.}:
    \begin{itemize}
        \item[1)] {$V^{\query}_{\ell} = \Span\{ h_{\ell}\}$, where $h_{\ell}$ is a unit in $A,$}
        \item[2)] {$\dim(\sum_{\ell=1}^L V^{\info}_{\ell})= \sum_{\ell=1}^L\dim(V^{\info}_{\ell}),  $ where $V_{\ell}^{\info}:=V_{\ell}^{\enco}\cdot V_{\ell}^{\query}=\Span\{f\cdot g:\; f\in V_{\ell}^{\enco},g\in V_{\ell}^{\query}\},$}
        \item[3)] {$V^{\info} \cap V^{\noise} =\{0\},$ where $V^{\info}:=\bigoplus_{\ell=1}^LV_{\ell}^{\info}$ and $V^{\noise}:=\sum_{\ell=1}^L(V_{\ell}^{\enco}{\cdot} V_{\ell}^{\priv}+V_{\ell}^{\secu}{\cdot} V_{\ell}^{\query}+V_{\ell}^{\secu}{\cdot} V_{\ell}^{\priv}),$}
        \item[4)] {$\varphi$ is injective over $V^{\info} \oplus V^{\noise} .$}
    \end{itemize}   
 Then there exists an $X$-secure and $T$-private information retrieval scheme with rate $\Rcal=\frac{L}{N}$, for any $X,T$ satisfying $$X\leq \min\{d (\varphi(V_{\ell}^{\secu})^\perp):\ell \in [L]\}-1,\; T\leq \min\{d (\varphi(V_{\ell}^{\priv})^\perp):\;\ell \in [L]\}-1.$$
\end{lemma}  
In the following, we recall the XSTPIR schemes derived based on the above general framework Lemma~\ref{lem:Jia_Makk_framework}, including those based on the rational curves (Theorem~\ref{thm:Jia_makk_rational}), the elliptic curves (Theorem~\ref{thm:makk_elliptic}), the hyperelliptic curves (Theorem~\ref{thm:makk_hyperelliptic}), and the Hermitian curves (Theorem~\ref{thm:GGMT_Hermitian}).

\begin{theorem}[{\cite[Section \Rmnum{6}]{jia2019cross} and \cite[Section \Rmnum{4}.A]{makkonen2024algebraic}}\footnote{In \cite[Section \Rmnum{6}]{jia2019cross}, this scheme was first proposed using Reed-Solomon code. It was subsequently unified under the framework of Lemma~\ref{lem:Jia_Makk_framework} and re-derived using Algebraic-Geometry codes over a rational curve.}]\label{thm:Jia_makk_rational}
    Let $q$ be a prime power. Let $L,X,T$ be three positive integers. If $q\geq 2L+X+T$, then there exists an $X$-secure and $T$-private information retrieval scheme over $\F_q$, with rate 
    $$
    \Rcal=\frac{L}{N}=1-\frac{X+T}{N},
    $$
    where $N=L+X+T.$
\end{theorem}
\begin{theorem}[{\cite[Section \Rmnum{5}]{makkonen2024algebraic}}] \label{thm:makk_elliptic}
    Let $\mathcal{E}$ be an elliptic curve defined over $\F_q$, $p=\Char(\F_q)>3$, by the affine equation $y^2-f(x)=0$ where $\deg(f)=3$, and let $X$ and $T$ be some fixed security and privacy parameters. Let $L$ be an odd integer, $J=\frac{L+1}{2}$, and $\{P_1,\bar {P_1},\dots,P_{J}, \bar {P_J}\}$ be a set of $2J$ $\F_q$-rational points of $\mathcal{E}$ distinct from $P_\infty$ and not lying on the line $y=0,$ where  $P_j=(\alpha_j,\beta_j)$  and $\bar P_j=(\alpha_j,-\beta_j),$  $j\in [J].$
    If there exist $N+1=L+X+T+9$ $\F_q$-rational points of $\mathcal{E}$ distinct from $P_{\infty},P_1,\ldots,P_J, \bar{P}_1,\ldots, \bar{P}_J$ and not lying on $y=0$, 
    then there exists an $X$-secure and $T$-private PIR scheme with rate $$\Rcal^{\mathcal{E}}=\frac{L}{N}=1-\frac{X+T+8}{N}.$$ 
\end{theorem}
As observed in \cite{makkonen2024algebraic}, even if $\Rcal^\mathcal{X}>\Rcal^\Ecal$ for fixed $X, T, N$, the above constructions might have distinct requirements for different field sizes. This follows from the fact that elliptic curves allow for more rational points compared to curves of genus 0. This opens up the possibility of leveraging a larger number of servers and enlarging the rate of the genus-1 scheme when operating with fixed field size $q$, security parameter $X$, and privacy parameter $T$. Actually, in \cite[Fig.1]{makkonen2024algebraic}, Makkonen \textit{et al.} show that for fixed field size and sufficiently large $X,T$, the XSTPIR schemes based on elliptic curves offer higher rates than those based on rational curves. Recently, this comparison is quantified in \cite[Proposition 4.2]{ghiandoni2025agcodes}.

Extending the results of \cite{makkonen2024algebraic}, XSTPIR schemes based on hyperelliptic curves have been proposed in \cite{makkonen2024secretsharingsecureprivate} as follows.
\begin{theorem}[{\cite[Theorem VI.1]{makkonen2024secretsharingsecureprivate}}]\label{thm:makk_hyperelliptic}
    Let $\Ycal$  be a hyperelliptic curve of genus $g\geq 1$ defined over $\mathbb{F}_q$ by the affine equation $y^2 +H(x)y = F(x)$.
Let $L \geq g,$  $L \equiv  g$ (mod 2) and $J =
\frac{L+g}{2}.$  Let $\gamma_1,\dots, \gamma_J \in \F_q$ \footnote{$\gamma_1,\dots,\gamma_J$ need to be pairwise distinct, although not explicitly stated in \cite[Theorem VI.1]{makkonen2024secretsharingsecureprivate}, see \cite[Section \Rmnum{6}.A]{makkonen2024secretsharingsecureprivate}} be such that $F(\gamma_j) \neq 0$ and set $h =
\prod_{j\in[J]}(x-\gamma_j).$
Let $N = L + X + T + 6g + 2$ for some security and privacy parameters $X$ and $T.$ If $\mathcal{Y}$ has $N + g$ rational
points $P_1,\dots, P_{N+g}$ different from $P_\infty$, not lying on $y=0$, and such that $h(P_n)\neq 0,$ then there exists an $X$-secure and $T$-private PIR scheme
with rate $$R^{\mathcal{Y}} = 1 -\frac{X + T + 6g + 2}{N}.$$ 
\end{theorem}
Recently, Ghiandoni \textit{et al.} \cite{ghiandoni2025agcodes} proposed the following XSTPIR schemes based on Hermitian curves. Compared with the above XSTPIR schemes based on rational curves, elliptic curves, and hyperelliptic curves, this scheme has a higher maximum rate especially when $X+T$ is large (see \cite[Section \Rmnum{4}]{ghiandoni2025agcodes} for details).
\begin{theorem}[{\cite[Theorem 3.3]{ghiandoni2025agcodes}}]\label{thm:GGMT_Hermitian}
    Let $\mathcal{H}_q$ be the Hermitian curve of genus $g(\mathcal{H}_ q)=\frac{q(q-1)}{2}.$ Let $g(\mathcal{H}_ q)\leq L\leq q^3-g(\mathcal{H}_ q),$ with $L+g(\mathcal{H}_ q)\equiv_q 0$ and set $m=\frac{L+g(\mathcal{H}_ q)}{q}.$ Let $P_{i,z}= (\alpha_i, \beta_{i,z})$, $i\in [m]$, $z\in [q]$, be $mq$ affine $\mathbb{F}_{q^2}$-rational points of $\mathcal{H}_q$ with $\alpha_i\neq 0$. If there exist $L+X+T+\frac{7q^2-3q-6}{2}+1$ $\F_{q^2}$-rational points of $\mathcal{H}_q$ distinct from $P_{\infty}, P_0$ and $P_{i,z}$, $i\in [m]$, $z\in [q]$, then there exists an $X$-secure and $T$-private PIR scheme with rate 
    \begin{equation}\label{eq:202510211415}
        \Rcal^{\Hcal_q} = \frac{L}{N}=1- \dfrac{X+T+3q^2-q-2}{N},
    \end{equation}
where $N=L+X+T+3q^2-q-2,$ and $X$ and $T$ are some fixed  security and privacy parameters.
\end{theorem} 
\subsection{Maximum Rates of Existing XSTPIR Schemes for Fixed Field Size and $X$, $T$}
In this subsection, we recall and present some results on the maximum rates of those existing XSTPIR schemes introduced in the previous subsection.
\begin{proposition}[{\cite[Proposition 4.1]{ghiandoni2025agcodes}}]\label{prop:Jia_makk_rational_maximum}
    The maximum rate of the XSTPIR scheme introduced in Theorem~\ref{thm:Jia_makk_rational} is 
    \begin{equation}\label{eq:jia_makk_rational_maximum}
    \Rcal_{\max}^\mathcal{X}(q,X,T)=\frac{\floor{\frac{q-X-T}{2}}}{\floor{\frac{q-X-T}{2}}+X+T},
    \end{equation}
    under the condition that 
    \begin{align}\label{eq:jia_makk_rational_condition}
        \floor{\frac{q-X-T}{2}}\geq 1.
    \end{align} 
 \end{proposition} 
 
\begin{proposition}[{\cite[Proposition 4.1]{ghiandoni2025agcodes}}]\label{prop:makk_elliptic_maximum}
   The maximum rate of the XSTPIR scheme described in Theorem~\ref{thm:makk_elliptic} is 
   \begin{align}\label{eq:202510281416}
   \Rcal_{\max}^\mathcal{\Ecal}(q,X,T)=\frac{L_{\max}^\mathcal{E}}{L_{\max}^\mathcal{E}+X+T+8},
   \end{align}
   where $L_{\max}^\mathcal{E}=2 \cdot \left\lfloor \frac{\# \Ecal(\F_q)-(X+T+\gamma+9)}{4} \right\rfloor-1$, $\gamma:=\#\{z\in \F_q:\;f(z)=0\}$.
\end{proposition}

In the following, an upper bound on the maximum rates of XSTPIR schemes based on hyperelliptic curves is provided. It is an adaptation of the upper bounds in \cite[Corollary 4.5, and the statement right after it]{ghiandoni2025agcodes}, with slightly relaxed conditions (the curve is of the general form $y^2+H(x)y=F(x)$, over finite fields of arbitrary characteristics). For completeness, a proof of it is also included.
\begin{proposition} \label{prop:hyperelliptic_maximum}
Let $\Ycal:y^2+H(x)y=F(x)$ be a hyperelliptic curve of genus $g$ defined over $\F_q$, where $\F_q$ is a prime power, $\deg(H(x))\leq g$, and $\deg(F(x))=2g+1$. Let $J,L,N,X,T$ be as in Theorem \ref{thm:makk_hyperelliptic}.
Then we have \begin{equation} \label{eq:new_bound Jmax per gamma eq 0}
    J_{\max}^\Ycal \le \left\lfloor\frac{2q-(X+T+6g+2)}{4}\right\rfloor,
\end{equation} 
where $J^{\Ycal}_{\max}$ denotes the maximum possible value of $J$ with respect to the curve $\Ycal$ and fixed $X,T$. 
In particular, 
\begin{align}\label{eq:202510281457}
\Rcal_{\max}^{\Ycal}(q,X,T)=\frac{2J_{\max}^{\Ycal}-g}{2J_{\max}^{\Ycal}+X+T+5g+2}\leq \frac{2q-(X+T+8g+2)}{2q+X+T+4g+2},
\end{align}
where $\Rcal^{\Ycal}_{\max}(q,X,T)$ denotes the maximum rate of all XSTPIR schemes based on the hyperelliptic curve $\Ycal$, and $q,X,T$ denote the fixed field size, security parameter, and privacy parameter, respectively.

Define $\Rcal_{\max}^g(q,X,T)$ as the maximum value of $\Rcal_{\max}^\Ycal(q,X,T)$ with $\Ycal$ ranging among all the hyperellptic curves of equation $y^2+H(x)y=F(x)$ of genus $g$. By \eqref{eq:202510281457}, We have 
\begin{align}\label{eq:makk_hyper_maximum_bound}
  \Rcal_{\max}^g(q,X,T) \leq  \frac{2q-(X+T+8g+2)}{2q+X+T+4g+2}\xlongequal{\text{(we denote it by)}}\overline{\Rcal^{g}_{\max}}(q,X,T).  
\end{align}
\end{proposition}
\begin{proof}
   Adopting the notation in Theorem~\ref{thm:makk_hyperelliptic}. Assume that there exist $J$ distinct elements $\gamma_1,\dots,\gamma_{J}$ such that there exist $N+g=2J-g+X+T+6g+2+g=2J+X+T+6g+2$ distinct rational points $P_{1},\dots,P_{2J+X+T+6g+2}$ distinct from $P_{\infty}$, not lying on $y=0$, and such that $h(P_{n})\neq 0$ for each $1\leq n\leq 2J+X+T+6g+2$, where $h(x)=(x-\gamma_1)\cdots(x-\gamma_{J})$. 
   Note that \begin{itemize}
       \item $x(P_n)\in \F_q\backslash \{\gamma_1,\dots,\gamma_J\}$ for each $1\leq n\leq 2J+X+T+6g+2$, since $h(P_{n})\neq 0$ for each $1\leq n\leq 2J+X+T+6g+2$;
       \item for any $\alpha\in \F_q\backslash \{\gamma_1,\dots,\gamma_J\}$, the equation of $y$: $y^2+H(\alpha)y=F(\alpha)$ have most $2$ distinct solutions in $\F_q$.
   \end{itemize} 
   We have $2(q-J)\geq 2J+X+T+6g+2$, which yields $J\leq \floor{\frac{2q-(X+T+6g+2)}{4}}$ as $J$ is an integer. Therefore, we have $J_{\max}^{\Ycal}\leq \floor{\frac{2q-(X+T+6g+2)}{4}}$. 

   From the rate formula in Theorem~\ref{thm:makk_hyperelliptic}, it can be observed that for a hyperelliptic curve $\Ycal$ and fixed $X$, $T$, the maximum rate is obtained by substituting $J_{\max}^{\Ycal}$ into the rate formula, i.e., $\Rcal_{\max}^{\Ycal}(q,X,T)=\frac{2J_{\max}^{\Ycal}-g}{2J_{\max}^{\Ycal}+X+T+5g+2}$. Since $J_{\max}^{\Ycal}\leq \floor{\frac{2q-(X+T+6g+2)}{4}}$, the inequalities \eqref{eq:202510281457} and \eqref{eq:makk_hyper_maximum_bound} hold.
\end{proof}
 
\begin{remark}\label{rem:202512191943}
    From \eqref{eq:makk_hyper_maximum_bound}, when discussing $\Rcal_{\max}^g(q,X,T)$ and $\overline{\Rcal^{g}_{\max}}(q,X,T)$, the implicit necessary condition 
    \begin{align}\label{eq:makk_hyper_condition}
     \overline{\Rcal^{g}_{\max}}(q,X,T)=\frac{2q-(X+T+8g+2)}{2q+X+T+4g+2}\geq 0
    \end{align}
     must be satisfied. Otherwise, the maximum rate is less than 0, which means that no feasible scheme can exist in this case.
\end{remark} 
\begin{remark}\label{rem:elliptic_hyperelliptic_maximum}
In Section~\ref{sec:4.hyper}, we will compare the maximum rates of our two new XSTPIR schemes based on rational curves and Hermitian curves with the maximum rates of the XSTPIR schemes based on elliptic curves (Theorem~\ref{thm:makk_elliptic} and Proposition~\ref{prop:makk_elliptic_maximum}) and hyperelliptic curves (Theorem~\ref{thm:makk_hyperelliptic}, Proposition~\ref{prop:hyperelliptic_maximum}), specifically by demonstrating that the maximum rate of our XSTPIR schemes exceeds the upper bound $\overline{\Rcal^{g}_{\max}}(q,X,T)=\frac{2q-(X+T+8g+2)}{2q+X+T+4g+2}$ in \eqref{eq:makk_hyper_maximum_bound}, under certain constraints. And we will no longer conduct a separate comparison with the elliptic curve scheme (Theorem~\ref{thm:makk_elliptic} and Proposition~\ref{prop:makk_elliptic_maximum}). The reason is as follows. By \eqref{eq:202510281416}, we have
\begin{align*}
\Rcal^{\Ecal}_{\max}(q,X,T)&= \frac{2 \cdot \left\lfloor \frac{\# \Ecal(\F_q)-(X+T+\gamma+9)}{4} \right\rfloor-1}{2 \cdot \left\lfloor \frac{\# \Ecal(\F_q)-(X+T+\gamma+9)}{4} \right\rfloor-1+X+T+8}\\
&\leq \frac{2 \cdot \frac{\floor{q+2\sqrt{q}+1}-(X+T+9)}{4}-1}{2 \cdot   \frac{\floor{q+2\sqrt{q}+1}-(X+T+9)}{4}-1+X+T+8}\\
&\leq \frac{2 \cdot \frac{2q+1-(X+T+9)}{4}}{2 \cdot  \frac{2q+1-(X+T+9)}{4}-1+X+T+8}\\
&=\overline{\Rcal^{g}_{\max}}(q,X,T)|_{g=1},
\end{align*}
where the first ``$\leq$'' is due to the Hasse-Weil bound.
This inequality means that, when $g=1$, the upper bound $\overline{\Rcal^{g}_{\max}}(q,X,T)=\frac{2q-(X+T+8g+2)}{2q+X+T+4g+2}$ in \eqref{eq:makk_hyper_maximum_bound} for the maximum rate of the XSTPIR scheme based on hyperelliptic curves (Theorem~\ref{thm:makk_hyperelliptic}, Proposition~\ref{prop:hyperelliptic_maximum}) is also a valid upper bound for the maximum rate of the XSTPIR scheme based on elliptic curve (Theorem~\ref{thm:makk_elliptic} and Proposition~\ref{prop:makk_elliptic_maximum}). 
Thus, it suffices to compare our scheme with the upper bound $\overline{\Rcal^{g}_{\max}}(q,X,T)=\frac{2q-(X+T+8g+2)}{2q+X+T+4g+2}$ in \eqref{eq:makk_hyper_maximum_bound}.
This is why we only compare the maximum rates of our schemes with the upper bound \eqref{eq:makk_hyper_maximum_bound}, without a specific comparison with schemes based on elliptic and hyperelliptic curves, as in upcoming Proposition~\ref{prop:202511222144} and Propositions~\ref{prop:202511222143} in Section~\ref{sec:4}.
\end{remark} 
The maximum rate of the XSTPIR scheme based on Hermitian curves proposed in \cite{ghiandoni2025agcodes} is as follows.
\begin{proposition}[{\cite[Proposition 3.5]{ghiandoni2025agcodes}}]\label{prop:GGMT_hermitian_maximum}
 \label{prop rate migliore Hermitiana}
    The maximum rate  of the XSTPIR schemes in Theorem \ref{thm:GGMT_Hermitian} is given by
    \begin{align*}
        \Rcal_{\max}^{\Hcal_q}(q^2,X,T)=\frac{L}{N},
    \end{align*}
    where $L=mq-g(\Hcal_q)$ ($g(\Hcal_q)=\frac{q(q-1)}{2}$), $m=\lfloor \frac{q^3-3q^2+q+1-X-T}{2q} \rfloor$, and $N=L+X+T+3q^2-q-2$.
    Explicitly,
    \begin{align}\label{eq:ghi_hermitian_maximum}
        \Rcal_{\max}^{\Hcal_q}(q^2,X,T)=\frac{q\lfloor \frac{q^3-3q^2+q+1-X-T}{2q} \rfloor -\frac{q(q-1)}{2}}{q\lfloor \frac{q^3-3q^2+q+1-X-T}{2q} \rfloor-\frac{q(q-1)}{2}+X+T+3q^2-q-2}.
    \end{align}
\end{proposition}
\begin{remark}\label{rem:implicit_condition_of_Herimitian_GGMT}
The Proposition~\ref{prop:GGMT_hermitian_maximum} implicitly assumes a condition 
\begin{align}\label{eq:ghi_hermitian_condition}
m=\floor{\frac{q^3-3q^2+q+1-X-T}{2q}}\geq q-1,
\end{align}
which is consistent with the range of $L=mq-g(\Hcal_q)$, as specified in Theorem~\ref{thm:GGMT_Hermitian},
$$g(\Hcal_q)\leq L\leq q^3-g(\Hcal_q).$$ 
\end{remark}
Later in Section~\ref{sec:4}, we will make comparisons between the maximum rates of our proposed schemes and all existing schemes. For ease of comparison, in Table~\ref{tab:allmaximumrates} below, we list all maximum rates (or upper bounds of maximum rates) that we will compare later in Section~\ref{sec:4}, including all schemes introduced above (Rows 2-4) and our upcoming schemes that will be presented in Section~\ref{sec:3} (Rows 5-6).
\begin{table}[H]
    \centering
    \begin{tabular}{|c|c|c|c|} \hline
    The Curves Based on and References & Maximum Rates for Fixed Field Sizes and $X,T$ & Implicit Conditions \\ \hline
    \makecell{Rational curves \\
    Theorem~\ref{thm:Jia_makk_rational} and Proposition~\ref{prop:Jia_makk_rational_maximum} 
    }
    & $\Rcal_{\max}^\mathcal{X}(q,X,T)=\frac{\floor{\frac{q-X-T}{2}}}{\floor{\frac{q-X-T}{2}}+X+T}$ (see \eqref{eq:jia_makk_rational_maximum}) & $\floor{\frac{q-X-T}{2}}\geq 1$ (see \eqref{eq:jia_makk_rational_condition})   \\ \hline 
    \makecell{Hyperelliptic curves\\Theorem~\ref{thm:makk_hyperelliptic} and Proposition~\ref{prop:hyperelliptic_maximum}}& $\overline{\Rcal^{g}_{\max}}(q,X,T)=\frac{2q-(X+T+8g+2)}{2q+X+T+4g+2}$ (see \eqref{eq:makk_hyper_maximum_bound})  & $\frac{2q-(X+T+8g+2)}{2q+X+T+4g+2}\geq 0$ (see \eqref{eq:makk_hyper_condition})  \\ \hline
    \makecell{Hermitian curves\\
    Theorem~\ref{thm:GGMT_Hermitian} and Proposition~\ref{prop rate migliore Hermitiana}}
    &\makecell{$\Rcal^{\Hcal_q}_{\max}(q^2,X,T)=$\\$\frac{q\lfloor \frac{q^3-3q^2+q+1-X-T}{2q} \rfloor -\frac{q(q-1)}{2}}{q\lfloor \frac{q^3-3q^2+q+1-X-T}{2q} \rfloor-\frac{q(q-1)}{2}+(X+T+3q^2-q-2)}$ (see~\eqref{eq:ghi_hermitian_maximum}) }&$\floor{\frac{q^3-3q^2+q+1-X-T}{2q}}\geq q-1$ (see \eqref{eq:ghi_hermitian_condition})\\\hline
    \makecell{Rational curves \\
        Theorem~\ref{thm:our_XSTPIR_Rational} and Corollary~\ref{cor:Our_rational_maximum}}
        & \makecell{$\Rcal^{\Xcal}_{\new,\max}(q,X,T)=$\\ $\frac{2\floor{\frac{q-X-T-3}{2}}}{2\floor{\frac{q-X-T-3}{2}}+X+T+2}$ (see \eqref{eq:our_rational_maximum})} & $\floor{\frac{q-X-T-3}{2}}\geq 1$ (see \eqref{eq:condition_our_rational}) \\ \hline
        \makecell{Hermitian curves\\Theorem~\ref{thm:our_XSTPIR_Hermitian} and Corollary~\ref{cor:Our_Hermitian_maximum}} & \makecell{$\Rcal_{\new,\max}^{\Hcal_q}(q^2,X,T)=$\\ $\frac{2q\floor{\frac{q^3-3q^2-3q+2-X-T}{2q}}-\frac{q(q-1)}{2}}{2q\floor{\frac{q^3-3q^2-3q+2-X-T}{2q}}-\frac{q(q-1)}{2}+(X+T+3q^2+2q-2)}$ (see \eqref{eq:our_hermitian_maximum}) } &  $\floor{\frac{q^3-3q^2-3q+2-X-T}{2q}}\geq \frac{q-1}{2}$ (see \eqref{eq:condition_our_hermitian}) \\ \hline
    \end{tabular}
    \caption{The maximum rates of existing XSTPIR schemes and our new XSTPIR schemes for fixed field size, security parameter $X$, and privacy parameter $T$}
    \label{tab:allmaximumrates}
\end{table}

\section{XSTPIR Schemes based on Rational Curves and Hermitian Curves}
In this section, we present our new XSTPIR schemes based on rational curves and Hermitian curves. Before proceeding, we give a family of bases of the linear spaces of the form $\Span_{\F_q}\{1,x,\dots,x^{k-1}\}$, which is key for the construction of our XSTPIR schemes.
\label{sec:3}  
\subsection{A Family of Bases for $\Span_{\F_q}\{1,x,\dots,x^{k-1}\}$}
\label{sec:3.1} 
For arbitrary $n$ distinct irreducible polynomials $f_1,f_2,\dots,f_n\in \F_q[x]$ over $\F_q$, we denote 
\begin{equation}\label{eq:202510111701}
    \Bcal(f_1,f_2,\dots,f_n):=\braces{x^j\frac{f_1f_2\dots f_n}{f_i}:\; 1\leq i\leq n\text{ and } 0\leq j\leq \deg(f_i)-1}.
\end{equation}
\begin{lemma}\label{lem:202510111619}
    Let $f_1,f_2,\dots,f_n\in \F_q[x]$ be $n$ distinct irreducible polynomials over $\F_q$. Then the set $\Bcal(f_1,f_2,\dots,f_n)$ defined in \eqref{eq:202510111701} is a basis of the space $\Span_{\F_q}\{1,x,\dots,x^{\deg(f_1)+\dots+\deg(f_n)-1}\}$. 
\end{lemma}
\begin{proof}
    Let $\overline{\F_q}$ be a fixed algebraic closure of $\F_q$. To prove this lemma, it suffices to prove that the elements of $\Bcal(f_1,f_2,\dots,f_n)$ are linearly independent over $\overline{\F_q}$. 

    Let $f_i(x)=\prod_{1\leq j\leq \deg(f_i)}(x-\alpha_{i,j})$, where $\alpha_{i,j}\in \overline{\F_q}$. 
    Since $\F_q$ is a perfect field and $f_i(x)$ are pairwise distinct irreducible polynomials over $\F_q$, the elements $\alpha_{i,j}\in\overline{\F_q}$ $(i\in [n],j\in [\deg(f_i)])$ are pairwise distinct. For each $i\in [n]$, the monomial basis $\{1,x,\dots,x^{\deg(f_i)-1}\}$ is equivalent (by an invertible $\overline{\F_q}$-linear transformation) to $\left\{\frac{f_i}{(x-\alpha_{i,j})}:\; 1\leq j\leq \deg(f_i)\right\}$ over $\overline{\F_q}$, as implied by the Lagrange interpolation theorem. 
    
    Thus, $\Bcal(f_1,f_2,\dots,f_{n})=\braces{x^j\frac{f_1f_2\dots f_n}{f_i}:\; 1\leq i\leq n\text{ and } 0\leq j\leq \deg(f_i)-1}$ is equivalent (by an invertible $\overline{\F_q}$-linear transformation) to 
    $\braces{\frac{f_i}{(x-\alpha_{i,j})}\frac{f_1f_2\dots f_n}{f_i}:\; 1\leq i\leq n\text{ and } 1\leq j\leq \deg(f_i)}$, whose elements are linearly independent over $\overline{\F_q}$ as implied by the Lagrange interpolation theorem. Therefore, the elements of $\Bcal(f_1,f_2,\dots,f_{n})$ are linearly independent over $\F_q$, and consequently, $\Bcal(f_1,f_2,\dots,f_{n})$ is a basis of $\Span_{\F_q}\{1,x,\dots,x^{\deg(f_1)+\dots+\deg(f_n)-1}\}$.
\end{proof}
\begin{remark} 
    The basis for the space $\Span_{\F_q}\{1,x,\dots,x^{k-1}\}$ is fundamental to the construction of XSTPIR schemes based on Lemma~\ref{lem:Jia_Makk_framework}. In previous works \cite{makkonen2024algebraic}, \cite{makkonen2024secretsharingsecureprivate}, \cite{ghiandoni2025agcodes}, the bases employed are the Lagrange bases (or their variations).
    
  Taking the construction in \cite[Section~\Rmnum{4}.B and \Rmnum{4}.C]{makkonen2024algebraic} using rational curves (function fields) as an example. Over there, a family of key rational functions $h_1,\dots,h_{L}$ are defined to be $h_{\ell}=\frac{\prod_{j\in [L]\backslash \{\ell\}}(x-\alpha_j)}{\prod_{j\in [L]}(x-\alpha_j)}=\frac{1}{x-\alpha_\ell}$ ($\ell\in [L]$), where $\alpha_1,\dots,\alpha_L\in \F_q$ are pairwise distinct. $h_1,\dots,h_{L}$ have rational places $(x-\alpha_\ell)_{0}$ ($\ell\in[L]$) as poles. These $L$ rational places can no longer be used for evaluation. 
  The new basis we introduced in the above lemma can be applied to reduce the number of such rational poles, which is key to obtaining XSTPIR schemes with higher maximum rates.
\end{remark}
\subsection{A New XSTPIR Scheme based on Rational Curves}
\label{sec:3.2}
In this subsection, we present our XSTPIR scheme based on rational curves.
\begin{theorem}\label{thm:our_XSTPIR_Rational}
    Let $q$ be a prime power. Let $L,X,T$ be three positive integers. If $q\geq L+X+T+3$ and $2\mid L$, then there exists an XSTPIR scheme over $\F_q$, with rate 
    \begin{align}\label{eq:202511021443}
    \Rcal_{\new}^{\Xcal}=\frac{L}{N},\text{ where } N=L+X+T+2.
    \end{align}
\end{theorem}
\begin{proof}
Let $\Xcal$ be a rational curve defined over $\F_q$. Its function field, which is a rational function field, is denoted by $F=\F_q(x)$.
    Let $f_1,\dots,f_{\frac{L}{2}}\in \F_q[x]$ be $\frac{L}{2}$ distinct quadratic irreducible polynomial over $\F_q$. 
    
Define $h:=\frac{1}{f_1f_2\cdots f_{\frac{L}{2}}},$ and
\begin{equation}\label{eq:202510111502}
(h_{1},h_{2},\dots,h_{L}):=h\cdot\left(\frac{f_1f_2\cdots f_{\frac{L}{2}}}{f_1}, \frac{x\cdot f_1f_2\cdots f_{\frac{L}{2}}}{f_1},\dots,\frac{f_1f_2\cdots f_{\frac{L}{2}}}{f_{\frac{L}{2}}}, \frac{x\cdot f_1f_2\cdots f_{\frac{L}{2}}}{f_{\frac{L}{2}}}\right)
=\left(\frac{1}{f_1}, \frac{x}{f_1},\dots,\frac{1}{f_{\frac{L}{2}}}, \frac{x}{f_{\frac{L}{2}}}\right).
\end{equation} 
For each $a\in \brackets{1,\frac{L}{2}}$, we have
\begin{equation}\label{eq:202510161545}
   (h_{2a-1})=2P_{\infty}-(f_a)_0, \text{ and }  (h_{2a})=(x)_0+P_{\infty}-(f_a)_0.
\end{equation} 
For each $\ell\in [L]$, we define
    \begin{equation}\label{eq:202510052156}
        \begin{tabular}{ll}
$V^{\enco}_{\ell}=\Span\{1\},$ & $V^{\secu}_{\ell}=\frac{1}{h_{\ell}}\cdot \Lcal((X-1)P_{\infty})$, \\ 
        $V^{\query}_{\ell}=\Span\{h_{\ell}\},$ & $V^{\priv}_{\ell}=\Lcal((T-1)P_{\infty}).$
    \end{tabular}
    \end{equation}  

By Lemma~\ref{lem:202510111619} and the second term of \eqref{eq:202510111502}, $h_1,\dots,h_{L}$ are linearly independent over $\F_q$. Thus, $V^{\info}=\sum_{\ell\in [L]} V^{\enco}V^{\query}=\sum_{\ell\in [L]} \Span\{h_{\ell}\}$ is a direct sum. The condition 2) of Lemma~\ref{lem:Jia_Makk_framework} is satisfied. By \eqref{eq:202510161545}, we have
\begin{align}\label{eq:202601061645}
    \sum_{\ell=1}^{L}V^{\enco}_{\ell}V^{\priv}&=\Lcal((T-1)P_{\infty}),\\  
    \sum_{\ell=1}^{L}V^{\secu}_{\ell}V^{\query}_{\ell}&=\Lcal((X-1)P_{\infty}),\\ \notag
    \sum_{\ell=1}^{L}V^{\secu}_{\ell}V^{\priv}_{\ell}&\subseteq \sum_{\ell=1}^{L}\frac{1}{h_{\ell}}\Lcal((X-1+T-1)P_{\infty})\\  
    &\subseteq \Lcal((X+T-2)P_{\infty}+2P_{\infty}+(x)_0). 
\end{align} 
Therefore, we have
\begin{align}\label{eq:202510131621}
    V^{\info}=\sum_{\ell\in[L]}V^{\enco}_{\ell}V^{\query}_{\ell}=\sum_{\ell\in[L]}\Span\{h_{\ell}\}\subseteq \Lcal\parentheses{(f_1)_0+\cdots+(f_{\frac{L}{2}})_0-P_{\infty}}
\end{align}
and 
\begin{align}\label{eq:202510131624}
    V^{\noise}=\sum_{\ell\in [L]}(V^{\enco}_{\ell}V^{\priv}_{\ell}+V^{\secu}_{\ell}V^{\query}_{\ell}+V^{\secu}_{\ell}V^{\priv}_{\ell})\subseteq \Lcal((X+T)P_{\infty}+(x)_0).
\end{align}
By \eqref{eq:202510131621} and \eqref{eq:202510131624}, any nonzero element of $V^{\info}$ has at least one pole in $\Supp\parentheses{(f_1)_0+\cdots+(f_{\frac{L}{2}})_0}$, while any nonzero element of $V^{\noise}$ has no pole in $\Supp\parentheses{(f_1)_0+\cdots+(f_{\frac{L}{2}})_0}$. Thus, we have $V^{\info}\cap V^{\noise}=\{0\}$. The condition 3) of Lemma~\ref{lem:Jia_Makk_framework} is satisfied.
Also by \eqref{eq:202510131621} and \eqref{eq:202510131624}, we have $V^{\info}\oplus V^{\noise}\subseteq \Lcal(D^{\full})$, where
\begin{align}
  D^{\full}=(f_1)_0+\dots+(f_{\frac{L}{2}})_0+\big(X+T\big)P_{\infty}+(x)_0.   
\end{align} 
A direct calculation yields $\deg(D^{\full})=L+X+T+1$.

 Since $\#(\Supp(D^{\full})\cap \Pbb_F^1)=\#(\{P_{\infty}\}\cup \Supp((x)_0))=2$ and $\#\Pbb_{F}^1=q+1\geq (L+X+T+2)+2$, we can choose a set $\Pcal \subseteq \Pbb_F^1\backslash \Supp(D^{\full})$ with $\#\Pcal=N=\parentheses{L+X+T+2}=\deg(D^{\full})+1$. 

 Recall that $V^{\info}\oplus V^{\noise}\subseteq \Lcal(D^{\full})$. By Lemma~\ref{lem:Jia_Makk_framework}, to prove this theorem, it suffices to show that there exists an $\F_{q}$-algebra homomorphism $\vp:\; A \rightarrow \F_{q}^N$ ($A$ and $\varphi$ will be defined later) such that  
\begin{itemize}
    \item $\vp$ is injective over $\Lcal(D^{\full})$, 
    \item $h_{\ell}$ is invertible in $A$ for each $\ell\in [L]$,
    \item $d(\vp(V^{\secu}_{\ell})^{\perp})\geq X+1$ and $d(\vp(V^{\priv}_{\ell})^{\perp})\geq T+1$ for each $\ell\in [L]$.
\end{itemize}
 
 By Lemma~\ref{lem:AGcode_lem1}, since $\#\Pcal>L+X+T+1=\deg(D^{\full})$, the associated AG code $\Ccal(\Pcal,D^{\full})$ is of dimension $\ell(D^{\full})=\deg(D^{\full})+1-g(\Xcal)=L+X+T+2$, and the evaluation map $\vp:\; \Lcal(D^{\full})\rightarrow \F_{q}^N$ is injective by Lemma~\ref{lem:AGcode_lem1}. Extending the domain of $\vp$ from $\Lcal(D^{\full})$ to the holomorphy ring $A:=\bigcap\limits_{P\in \Pcal}\Ocal_{P}=\{f\in F:\; v_{P}(f)\geq 0 \text{ for all } P\in \Pcal\}$, which is an $\F_q$-algebra. The map $\vp$ becomes an $\F_{q}$-algebra homomorphism that is injective over $\Lcal(D^{\full})$. Also, it is direct to verify that $h_{\ell}$ is invertible in $A$ for each $\ell\in [L]$, by \eqref{eq:202510161545}. 

 Finally, we prove that $d(\vp(V^{\secu}_{\ell})^{\perp})\geq X+1$ and $d(\vp(V^{\priv}_{\ell})^{\perp})\geq T+1$ for each $\ell\in [L]$. Note that $V^{\secu}_{\ell}=\frac{1}{h_{\ell}}\cdot \Lcal((X-1)P_{\infty})=\Lcal\parentheses{(X-1)P_{\infty}-\parentheses{\frac{1}{h_{\ell}}}}$ and $V^{\priv}_{\ell}=\Lcal((T-1)P_{\infty}$. Since 
 \begin{align}
     &\Supp\parentheses{(X-1)P_{\infty}-\parentheses{\frac{1}{h_{\ell}}}}\cap \Pcal=\varnothing,\\
     &\Supp((T-1)P_{\infty})\cap \Pcal=\varnothing,\\
     &\# \Pcal-2=N-2=L+X+T\geq \deg\parentheses{(X-1)P_{\infty}-\parentheses{\frac{1}{h_{\ell}}}},\\
     &\# \Pcal-2=N-2=L+X+T\geq \deg((T-1)P_{\infty})),
 \end{align}
by Lemma~\ref{lem:AGcode_lem2}, the linear codes $\vp(V^{\secu}_{\ell})^{\perp}$ and $\vp(V^{\priv}_{\ell})^{\perp}$ each has dimension at least 1, and
\begin{align}
    d(\vp(V^{\secu}_{\ell})^{\perp})&\geq \deg\left((X-1)P_{\infty}-\left(\frac{1}{h_{\ell}}\right)\right)-(2g(\Xcal)-2)=X+1,\\
    d(\vp(V^{\priv}_{\ell})^{\perp})&\geq \deg((T-1)P_{\infty})-(2g(\Xcal)-2)=T+1.
\end{align}
By Lemma~\ref{lem:Jia_Makk_framework}, the proof is completed.
\end{proof} 
\begin{remark}  
The main difference between the above scheme (described in Theorem~\ref{thm:our_XSTPIR_Rational} and its proof) and the scheme in Theorem~\ref{thm:Jia_makk_rational} (described detailedly in \cite[Section~\Rmnum{4}.B]{makkonen2024algebraic}) is the selection of $h_{\ell}$ ($\ell\in [L]$). In \cite[Section~\Rmnum{4}.B]{makkonen2024algebraic}, $h_\ell=\frac{1}{x-\alpha_\ell}$ ($\ell \in [L]$), where $\alpha_{\ell}$ ($\ell \in [L]$) are pairwise distinct. Our $h_{\ell}$ ($\ell\in [L]$) are given in \eqref{eq:202510111502}. This adjustment of $h_{\ell}$ makes $\deg(D^{\full})$ from $L+X+T-1$ to $L+X+T+1$, and makes $\#(\Supp(D^{\full})\cap \Pbb_F^1)$ from $L+1$ to $2$. This leads to the improvement of maximum rates in some cases. See Proposition~\ref{prop:202511222202} for a detailed comparison.
\end{remark} 

\subsection{A New XSTPIR Scheme based on Hermitian Curves}
\label{sec:3.3}
In this subsection, we present our XSTPIR scheme based on Hermitian curves.
\begin{theorem}\label{thm:our_XSTPIR_Hermitian}
    Let $q$ be a prime power. 
    Let $L,X,T$ be three positive integers, where $L=mq-\frac{q(q-1)}{2}$ for an even integer $m\in [q-1,q^2-1]$. 
    Assume $$q^3+1-(q+1)\geq L+X+T+\frac{7q^2+3q-6}{2}+1
    .$$ 
    Then there exists an XSTPIR scheme over $\F_{q^2}$ with rate 
    \begin{align}\label{eq:202510311700}
    \Rcal_{\new}^{\Hcal_q}=\frac{L}{N}=1-\frac{X+T+3q^2+2q-2}{N},\text{ where } N=L+X+T+3q^2+2q-2.
    \end{align}
\end{theorem}
\begin{proof}
Let $\Hcal_q/\F_{q^2}$ be the Hermitian curve, and let $F/\F_{q^2}$ be its function field, where $F=\F_{q^2}(x,y)$ with transcendental elements $x,y$ satisfying $x^{q+1}=y^q+y$. They are of genus $g=\frac{q(q-1)}{2}$. 

Let $f_1,\dots,f_{\frac{m}{2}}\in \F_{q^2}[x]$ be $\frac{m}{2}$ pairwise distinct quadratic irreducible polynomial. This is valid since there are totally $\frac{q^4-q^2}{2}$ quadratic irreducible polynomial over $\F_{q^2}$, and $\frac{q^4-q^2}{2}\geq \frac{q^2-1}{2}\geq \frac{m}{2}$.
Denote $f_0:=x$.
    For $z\in [q]$, we define 
    \begin{align}\notag
        &\{h^{(z)}_{1},\dots,h^{(z)}_{m-z+1}\}:=\\ \label{eq:202510091433}
        &\begin{cases} 
            \left\{y^{z-1}\cdot x^j\frac{f_{1}\cdots f_{\frac{m-z+1}{2}}}{f_i}:\;1\leq i\leq \frac{m-z+1}{2},0\leq j\leq \deg(f_i)-1\right\}=y^{z-1}\cdot \Bcal(f_1,\dots,f_{\frac{m-z+1}{2}}),&\text{if }2\mid (m-z+1),\\ 
            \left\{y^{z-1}\cdot x^j\frac{f_{0}\cdots f_{\frac{m-z}{2}}}{f_i}:\;0\leq i\leq \frac{m-z}{2},0\leq j\leq \deg(f_i)-1\right\}=y^{z-1}\cdot \Bcal(f_0,f_1,\dots,f_{\frac{m-z}{2}}),&\text{if }2\mid (m-z),
        \end{cases}
    \end{align}
    where the symbol $\Bcal(\cdots)$ is defined as in \eqref{eq:202510111701}.
In the above, we have defined $\sum_{z=1}^{q} (m-z+1)=\frac{q}{2}(2m-q+1)=mq-\frac{q(q-1)}{2}=L$ distinct functions. Let 
    $h:=\frac{1}{f_1f_2\dots f_{\frac{m}{2}}}.$ We define 
\begin{align}\label{eq:202512192230}
\{h_1,\dots,h_{L}\}:=\braces{h\cdot h^{(z)}_i:\; 1\leq z\leq q,1\leq i\leq m-z+1},
\end{align}
where the order is immaterial. We claim that $h_1,\dots,h_{L}$ are linearly independent over $\F_{q^2}$. To prove this claim, we only need to prove that $\braces{h^{(z)}_i:\; 1\leq z\leq q,1\leq i\leq m-z+1}$ are linearly independent over $\F_{q^2}$. 
Indeed, by Lemma~\ref{lem:202510111619} and \eqref{eq:202510091433}, we can transform $\braces{h^{(z)}_i:\; 1\leq z\leq q,1\leq i\leq m-z+1}$ into $\braces{y^{z-1}x^{i}:\; 1\leq z\leq q,0\leq i\leq m-z}$ by an invertible $\F_{q^2}$-linear transformation. By the strict triangle inequality (see \cite[Lemma 1.1.11]{stichtenoth2009algebraic}), to prove that they are linearly independent over $\F_{q^2}$, it suffices to show that $v_{P_\infty}(y^{z-1}x^{i})$ ($1\leq z\leq q,0\leq i\leq m-z$) are pairwise distinct, where $P_{\infty}$ is the place at infinity of the Hermitian function field $F$. Assume $v_{P_\infty}(y^{z-1}x^{i})=v_{P_\infty}(y^{z'-1}x^{i'})$ for some $1\leq z,z'\leq q,0\leq i,i'\leq m-z$, then we have $-(z-1)(q+1)-iq=-(z'-1)(q+1)-i'q$, and then $(z-z')(q+1)=(i'-i)q$. Since $q$ is coprime to $q+1$, we have $q\mid (z-z')$, which implies $z=z'$ by the fact that $-(q-1)\leq z-z'\leq q-1$. Thus, we have $i=i'$. This claim is proved. 

For each $\ell\in [L]$, we define 
\begin{equation}\label{eq:202510091551}
        \begin{tabular}{ll}
        $V^{\enco}_{\ell}=\Span\{1\},$ & $V^{\secu}_{\ell}=\frac{1}{h_{\ell}}\cdot \Lcal((X+2g-1)P_{\infty})$,\\ 
        $V^{\query}_{\ell}=\Span\{h_{\ell}\},$ & $V^{\priv}_{\ell}=\Lcal((T+2g-1)P_{\infty}).$
    \end{tabular}
    \end{equation}  
By \eqref{eq:202510091433} and \eqref{eq:202512192230}, we have $v_{P_\infty}(h_{\ell})\geq 1$ and $(h_\ell)_{\infty}\leq (f_1)_0+\cdots+(f_{\frac{m}{2}})_0$ for each $\ell\in [L]$. Therefore,
\begin{align}\label{eq:202510091648}
    \sum_{\ell=1}^{L} V^{\enco}_{\ell}V^{\query}_{\ell}=\Span\{h_{\ell}:\;1\leq \ell\leq L\}\subseteq \Lcal((f_1)_0+\dots+(f_{\frac{m}{2}})_0-P_{\infty}).
\end{align}
Moreover, the sum above is a direct sum since $h_1,\dots,h_{L}$ are $\F_{q^2}$-linearly independent as previously established. The condition 2) of Lemma~\ref{lem:Jia_Makk_framework} is satisfied.

By \eqref{eq:202510091433} and \eqref{eq:202512192230}, we have $\parentheses{\frac{1}{h_{\ell}}}_{\infty}=\parentheses{h_\ell}_{0}\leq 2qP_{\infty}+(x^2)_0+(y^{q-1})_0$ for each $\ell\in[L]$. Therefore, we have 
\begin{align}\label{eq:202510101648}
    \sum_{\ell=1}^{L}V^{\enco}_{\ell}V^{\priv}&=\Lcal((T+2g-1)P_{\infty}),\\ \label{eq:202510101702}
    \sum_{\ell=1}^{L}V^{\secu}_{\ell}V^{\query}_{\ell}&=\Lcal((X+2g-1)P_{\infty}),\\ \notag
    \sum_{\ell=1}^{L}V^{\secu}_{\ell}V^{\priv}_{\ell}&\subseteq \sum_{\ell=1}^{L}\frac{1}{h_{\ell}}\Lcal((X+2g-1+T+2g-1)P_{\infty})\\ \label{eq:202510101703}
    &\subseteq \Lcal((X+2g-1+T+2g-1)P_{\infty}+2qP_{\infty}+2(x)_0+(q-1)(y)_0). 
\end{align} 
Thus, by \eqref{eq:202510091648}, \eqref{eq:202510101648}, \eqref{eq:202510101702}, \eqref{eq:202510101703} we have
\begin{align}\label{eq:202510132105}
    V^{\info}&=\sum_{\ell=1}^{L} V^{\enco}_{\ell}V^{\query}_{\ell} \subseteq \Lcal((f_1)_0+\dots+(f_{\frac{m}{2}})_0-P_{\infty}),\\ \label{eq:202510132335}
    V^{\noise}&=\sum_{\ell=1}^{L}(V^{\enco}_{\ell}V^{\priv}+V^{\secu}_{\ell}V^{\query}_{\ell}+V^{\secu}_{\ell}V^{\priv}_{\ell})\subseteq \Lcal((X+T+2q+4g-2)P_{\infty}+2(x)_0+(q-1)(y)_0).
\end{align}
Thus, it holds that $V^{\info}\cap V^{\noise}=\{0\}$ since any nonzero element of $V^{\info}$ has at least one pole in $\Supp((f_1)_0+\dots+(f_\frac{m}{2})_0)$ while any element of $V^{\noise}$ has no pole in $\Supp((f_1)_0+\dots+(f_\frac{m}{2})_0)$. The condition 3) of Lemma~\ref{lem:Jia_Makk_framework} is satisfied.

By \eqref{eq:202510132105} and \eqref{eq:202510132335}, we have $V^{\info}\oplus V^{\noise}\subseteq \Lcal(D^{\full})$, where
\begin{align}
    D^{\full}=(f_1)_0+\dots+(f_{\frac{m}{2}})_0+\big(X+T+2q+4g-2\big)P_{\infty}+2(x)_0+(q-1)(y)_0.
\end{align} 
A direct calculation yields $\deg(D^{\full})=L+X+T+\frac{7q^2+3q-6}{2}$.
 
 Recall that $V^{\info}\oplus V^{\noise}\subseteq \Lcal(D^{\full})$. By Lemma~\ref{lem:Jia_Makk_framework}, to prove this theorem, it suffices to show that there exists an $\F_{q^2}$-algebra homomorphism $\vp:\; A \rightarrow \F_{q^2}^N$ ($A$ and $\varphi$ will be defined later) such that  \begin{itemize}
    \item $\vp$ is injective over $\Lcal(D^{\full})$, 
    \item $h_{\ell}$ is invertible in $A$ for each $\ell\in[L]$,
    \item $d(\vp(V^{\secu}_{\ell})^{\perp})\geq X+1$ and $d(\vp(V^{\priv}_{\ell})^{\perp})\geq T+1$ for each $\ell \in[L]$.
\end{itemize} 
 
 By Lemma~\ref{lem:AGcode_lem1}, since $\#(\Supp(D^{\full})\cap \Pbb_F^1)=\#(\{P_{\infty}\}\cup \Supp((y)_0)\cup \Supp((x)_0))=q+1$ and $\#\Pbb_{F}^1=q^3+1\geq (L+X+T+\frac{7q^2+3q-6}{2}+1)+(q+1)$, we can choose a set $\overline{\Pcal} \subseteq \Pbb_F^1\backslash \Supp(D^{\full})$ with $\#\overline{\Pcal}=\parentheses{L+X+T+\frac{7q^2+3q-6}{2}+1}=\deg(D^{\full})+1$. Since $\#\overline{\Pcal}>L+X+T+\frac{7q^2+3q-6}{2}=\deg(D^{\full})$, the associated AG code $\Ccal(\overline{\Pcal},D^{\full})$ is of dimension $\ell(D^{\full})=\deg(D^{\full})+1-g=L+X+T+3q^2+2q-2$, and the evaluation map $\overline{\vp}:\; \Lcal(D^{\full})\rightarrow \F_{q^2}^{\# \overline{\Pcal}}$ is injective. Then we can choose an information set $\Pcal\subseteq \overline{\Pcal}$ of the code $\Ccal(\overline{\Pcal},D^{\full})$ with $\#\Pcal=N=\ell(D^{\full})=L+X+T+3q^2+2q-2$, such that the evaluation map $\vp:\;\Lcal(D^{\full})\rightarrow \F_{q^2}^N$ is also injective. Extending the domain of $\vp$ from $\Lcal(D^{\full})$ to the holomorphy ring $A:=\bigcap\limits_{P\in \Pcal}\Ocal_{P}=\{f\in F:\; v_{P}(f)\geq 0 \text{ for all } P\in \Pcal\}$, which is an $\F_{q^2}$-algebra. The map $\vp$ becomes an $\F_{q^2}$-algebra homomorphism that is injective over $\Lcal(D^{\full})$. And it is direct to verify that $h_\ell$ is invertible in $A$ for each $\ell\in [L]$.

 Finally, we prove that $d(\vp(V^{\secu}_{\ell})^{\perp})\geq X+1$ and $d(\vp(V^{\priv}_{\ell})^{\perp})\geq T+1$ for each $\ell \in[L]$. Note that $V^{\secu}_{\ell}=\frac{1}{h_{\ell}}\cdot \Lcal((X+2g-1)P_{\infty})=\Lcal\parentheses{(X+2g-1)P_{\infty}-\parentheses{\frac{1}{h_{\ell}}}}$ and $V^{\priv}_{\ell}=\Lcal((T+2g-1)P_{\infty})$. Since 
 \begin{align}
     &\Supp\parentheses{(X+2g-1)P_{\infty}-\parentheses{\frac{1}{h_{\ell}}}}\cap \Pcal=\varnothing,\\
     &\Supp((T+2g-1)P_{\infty})\cap \Pcal=\varnothing,\\
     &\# \Pcal-2=N-2=L+X+T+3q^2+2q-2-2\geq  \deg\parentheses{(X+2g-1)P_{\infty}-\parentheses{\frac{1}{h_{\ell}}}},\\
     &\# \Pcal-2=N-2=L+X+T+3q^2+2q-2-2\geq \deg((T+2g-1)P_{\infty})),
 \end{align} 
by Lemma~\ref{lem:AGcode_lem2}, the linear codes $\vp(V^{\secu}_{\ell})^{\perp}$ and $\vp(V^{\priv}_{\ell})^{\perp}$ each has dimension at least 1, and
\begin{align}
    d(\vp(V^{\secu}_{\ell})^{\perp})&\geq \deg \parentheses{(X+2g-1)P_{\infty}-\parentheses{\frac{1}{h_{\ell}}}}-(2g-2)=X+1,\\
    d(\vp(V^{\priv}_{\ell})^{\perp})&\geq \deg((T+2g-1)P_{\infty})-(2g-2)=T+1.
\end{align}

By Lemma~\ref{lem:Jia_Makk_framework}, the proof is completed.
\end{proof}
\begin{remark}
   The main difference between the above scheme (described in Theorem~\ref{thm:our_XSTPIR_Hermitian} and its proof) and the scheme in Theorem~\ref{thm:GGMT_Hermitian} (described detailedly in \cite[Theorem 3.3]{ghiandoni2025agcodes}) is the selection of $h_{\ell}$ ($\ell\in [L]$). For their selection of $h_{\ell}$ ($\ell\in [L]$), see \cite[Corollary 3.2, and the sentence after (9)]{ghiandoni2025agcodes}. For our selection of $h_{\ell}$ ($\ell\in [L]$), see \eqref{eq:202510091433} and \eqref{eq:202512192230}. This adjustment makes $\deg(D^{\full})$ from $\deg(D^{\full})=L+X+T+\frac{7q^2-3q-6}{2}$ to $L+X+T+\frac{7q^2+3q-6}{2}$, and makes $\#(\Supp(D^{\full})\cap \Pbb_F^1)$ from $L+\frac{q(q-1)}{2}+2$ to $q+1$. This leads to the improvement of maximum rates. See Proposition~\ref{prop:202510191639} for a detailed comparison.
\end{remark}

\section{Rate Comparison}
\label{sec:4}
In the above section, we introduced our new XSTPIR schemes based on rational curves and Hermitian curves. In this section, let us make comparisons between the maximum rates of our new XSTPIR schemes with known constructions of XSTPIR schemes. The maximum rate of a family of XSTPIR schemes refers to the maximum achievable rate for fixed field size and $X,T$, with an unrestricted number of servers $N$.
Before comparison, we determine the maximum rates of our new XSTPIR scheme proposed in Theorems~\ref{thm:our_XSTPIR_Rational} and \ref{thm:our_XSTPIR_Hermitian}. 
\begin{corollary}\label{cor:Our_rational_maximum}
   Fixing the field size $q$, the security parameter $X\geq 1$, and the privacy parameter $T\geq 1$ in Theorem~\ref{thm:our_XSTPIR_Rational}, the maximum rate of the XSTPIR schemes described in Theorem~\ref{thm:our_XSTPIR_Rational} is given by 
    \begin{align}\label{eq:our_rational_maximum}
    \Rcal^{\Xcal}_{\new,\max}(q,X,T):=\frac{2\floor{\frac{q-X-T-3}{2}}}{2\floor{\frac{q-X-T-3}{2}}+X+T+2},
    \end{align}
    under the condition that 
    \begin{align}\label{eq:condition_our_rational}
    \floor{\frac{q-X-T-3}{2}}\geq 1.
    \end{align}
\end{corollary}
\begin{proof}
    By the rate expression \eqref{eq:202511021443} in Theorem~\ref{thm:our_XSTPIR_Rational}, maximizing the rate $\Rcal^{\Xcal}_{\new}=\frac{L}{N}=\frac{L}{L+X+T+2}$ for fixed $q,X,T$ is equivalent to finding the maximum $L$ such that
    $$q\geq L+X+T+3$$ where $L$ is a positive even integer. Note that the largest admissible $L$ is equal to $2\floor{\frac{q-3-X-T}{2}}$ (it is positive by the condition~\eqref{eq:condition_our_rational}). Substituting it into the expression $\Rcal^{\Xcal}_{\new}=\frac{L}{N}=\frac{L}{L+X+T+2}$, the proof is complete.
\end{proof}
\begin{corollary}\label{cor:Our_Hermitian_maximum}
   Fixing the field size $q^2$, the security parameter $X\geq 1$, and the privacy parameter $T\geq 1$ in Theorem~\ref{thm:our_XSTPIR_Hermitian}, the maximum rate of the XSTPIR schemes described in Theorem~\ref{thm:our_XSTPIR_Hermitian} is given by 
    \begin{align}\label{eq:our_hermitian_maximum}
    \Rcal_{\new,\max}^{\Hcal_q}(q^2,X,T):=\frac{2q\floor{\frac{q^3-3q^2-3q+2-X-T}{2q}}-\frac{q(q-1)}{2}}{2q\floor{\frac{q^3-3q^2-3q+2-X-T}{2q}}-\frac{q(q-1)}{2}+(X+T+3q^2+2q-2)},
    \end{align}
    under the condition that 
    \begin{align}\label{eq:condition_our_hermitian}
      \floor{\frac{q^3-3q^2-3q+2-X-T}{2q}}\geq \frac{q-1}{2}.  
    \end{align}    
\end{corollary}
\begin{proof}
    By the rate expression \eqref{eq:202510311700} in Theorem~\ref{thm:our_XSTPIR_Hermitian}, maximizing the rate $\Rcal^{\Hcal_q}_{\new}=\frac{L}{N}=\frac{L}{L+X+T+3q^2+2q-2}$ for fixed $q^2,X,T$ is equivalent to find the maximum $L$ such that
    $$q^3+1-(q+1)\geq L+X+T+\frac{7q^2+3q-6}{2}+1
    ,$$ where $L=mq-\frac{q(q-1)}{2}$ with $m\in[q-1,q^2-1]$ being an even integer. Thus, it suffices to find the maximum even integer $m$ such that 
    $$q^3+1-(q+1)\geq mq-\frac{q(q-1)}{2}+X+T+\frac{7q^2+3q-6}{2}+1
    .$$
    After rearranging the inequality, we find that the largest admissible $m$ is equal to $2\floor{\frac{q^3-3q^2-3q+2-X-T}{2q}}$ (it is in $[q-1,q^2-1]$ by the condition~\eqref{eq:condition_our_hermitian}). Substituting it into $\Rcal^{\Hcal_q}_{\new}=\frac{L}{N}=\frac{L}{L+X+T+3q^2+2q-2}$ ($L=mq-\frac{q(q-1)}{2}$), the proof is complete.
\end{proof}
Let us begin to compare the maximum rates of our XSTPIR schemes with existing ones.

\subsection{Comparison with the Existing XSTPIR Scheme based on Rational Curves}
\label{sec:4.rational}
In this subsection, we make comparisons between the maximum rates of our newly proposed XSTPIR schemes with existing XSTPIR schemes based on rational curves. 

In the following proposition, we make a comparison between the maximum rates of our proposed rational-curve-based XSTPIR scheme (Theorem~\ref{thm:our_XSTPIR_Rational}, Corollary~\ref{cor:Our_rational_maximum}, or see Row 5 of Table~\ref{tab:allmaximumrates}) and the rational-curve-based XSTPIR scheme by Jia \textit{et al.} \cite{jia2019cross} and Makkonen \textit{et al.} \cite{makkonen2024algebraic} (Theorem~\ref{thm:Jia_makk_rational}, Proposition~\ref{prop:Jia_makk_rational_maximum}, or see Row 2 of Table~\ref{tab:allmaximumrates}). 
\begin{proposition}\label{prop:202511222202}
Assume $4\leq X+T\leq q-10$ and $q\geq 21$, we have
    $$
    \Rcal_{\new,\max}^{\Xcal}(q,X,T)> \Rcal_{\max}^{\Xcal}(q,X,T).
      \footnote{\text{Under the assumption, the conditions~\eqref{eq:condition_our_rational} and \eqref{eq:jia_makk_rational_condition} are both satisfied, so we can make this comparison.}}
    $$ 
\end{proposition}
\begin{proof} 
For simplicity of comparison, we give estimations of $\Rcal_{\new,\max}^{\Xcal}(q,X,T)$ and $\Rcal_{\max}^{\Xcal}(q,X,T)$. By \eqref{eq:our_rational_maximum}, we have
    \begin{align}\label{eq:202510311622}
      \Rcal_{\new,\max}^{\Xcal}(q,X,T)\geq \frac{2\parentheses{\frac{q-X-T-3}{2}-\frac{1}{2}}}{2\parentheses{\frac{q-X-T-3}{2}-\frac{1}{2}}+X+T+2}= \frac{q-X-T-4}{q-2}\xlongequal{\text{(denoted by)}}\overline{\Rcal_{\new,\max}^{\Xcal}}(q,X,T).
    \end{align}
By \eqref{eq:jia_makk_rational_maximum}, we have
    \begin{align}\label{eq:202510311624}
      \Rcal^{\Xcal}_{\max}(q,X,T)\leq \frac{\frac{q-X-T}{2}}{\frac{q-X-T}{2}+X+T}=\frac{q-X-T}{q+X+T}\xlongequal{\text{(denoted by)}} \overline{\Rcal^{\Xcal}_{\max}}(q,X,T).
    \end{align}

    The difference of $\overline{\Rcal_{\new,\max}^{\Xcal}}(q,X,T)$ and $\overline{\Rcal^{\Xcal}_{\max}}(q,X,T)$ is equal to
    $$
      D(q,M):=\frac{(q-M-4)(q+M)-(q-M)(q-2)}{(q-2)(q+M)},
    $$
    where we use $M$ to denote $X+T$ for simplicity.
    Its numerator is 
    $$
    N(q,M):=(q-M-4)(q+M)-(q-M)(q-2)=-M^2-2q-6M+qM.
    $$
By \eqref{eq:202510311622} and \eqref{eq:202510311624}, to prove this proposition, it suffice to prove that 
$N(q,M)>0$ for any $4\leq M\leq q-10$ and $q\geq 21$.
Note that $N(q,M)$ is a concave function of $M$ for any fixed $q$, as it is a quadratic polynomial of $M$ with leading coefficient $-1$. It suffices to prove that $N(q,M)|_{M=4}=2q-40$ and $N(q,M)|_{M=q-10}=2q-40$ are both positive for any $q\geq 21$, which is straightforward to verify. The proof is complete.
\end{proof}
In the following proposition, we make a comparison between the maximum rates of our Hermitian-curve-based XSTPIR scheme (Theorem~\ref{thm:our_XSTPIR_Hermitian} and Corollary~\ref{cor:Our_Hermitian_maximum}, or see Row 6 of Table~\ref{tab:allmaximumrates}) and the rational-curve-based scheme by Jia \textit{et al.} \cite{jia2019cross} and Makkonen \textit{et al.} \cite{makkonen2024algebraic} (Theorem~\ref{thm:Jia_makk_rational}, Proposition~\ref{prop:Jia_makk_rational_maximum}, or see Row 2 of Table~\ref{tab:allmaximumrates}). Note that the field size in this proposition is $q^2$, rather than $q$. We need to replace $q$ with $q^2$ in the expression \eqref{eq:jia_makk_rational_maximum} of $\Rcal^{\Xcal}_{\max}(q,X,T)$ to get $\Rcal^{\Xcal}_{\max}(q^2,X,T)$. The $q$ in condition \eqref{eq:jia_makk_rational_condition} should also be replaced by $q^2$ accordingly. 
\begin{proposition}\label{prop:202511222203}
   Assume $3q\leq X+T\leq q^2-2$ and $q\geq 9$, we have
    $$
    \Rcal_{\new,\max}^{\Hcal_q}(q^2,X,T)>  \Rcal^{\Xcal}_{\max}(q^2,X,T). 
    \footnote{\text{Under the assumption, the conditions~\eqref{eq:condition_our_hermitian} and \eqref{eq:jia_makk_rational_condition} are both satisfied, so we can make comparison.}}
    $$ 
\end{proposition}
\begin{proof}
We begin with the following two estimations that will simplify the comparison. By \eqref{eq:our_hermitian_maximum},
   \begin{align}\notag
           \Rcal_{\new,\max}^{\Hcal_q}(q^2,X,T)&\geq \frac{q^3-3q^2-3q+2-X-T-(2q-1)-\frac{q(q-1)}{2}}{q^3-3q^2-3q+2-X-T-(2q-1)-\frac{q(q-1)}{2}+X+T+3q^2+2q-2} \\\notag
           &\geq   \frac{q^3-5q^2+3-X-T}{q^3-5q^2+3-X-T+X+T+3q^2+2q-2} \\\notag
           &=\frac{q^3-5q^2+3-X-T}{q^3-2q^2+2q+1}\\\label{eq:202510311559}
           &\xlongequal{\text{(denoted by)}}\overline{\Rcal_{\new,\max}^{\Hcal_q}}(q^2,X,T),
    \end{align}
where the first ``$\geq $'' is due to $2q\floor{\frac{q^3-3q^2-3q+2-X-T}{2q}}-\frac{q(q-1)}{2}\geq 2q\frac{q^3-3q^2-3q+2-X-T-(2q-1)}{2q}-\frac{q(q-1)}{2}$, and the second ``$\geq $'' is due to $q^3-3q^2-3q+2-X-T-(2q-1)-\frac{q(q-1)}{2}\geq q^3-5q^2+3-X-T$ for any $q\geq 3$.

By \eqref{eq:jia_makk_rational_maximum}, we have
     \begin{align}\label{eq:202510311600}
       \Rcal^{\Xcal}_{\max}(q^2,X,T)\leq \frac{\frac{q^2-X-T}{2}}{\frac{q^2-X-T}{2}+X+T}=\frac{q^2-X-T}{q^2+X+T}\xlongequal{\text{(denoted by)}} \overline{\Rcal^{\Xcal}_{\max}}(q^2,X,T).
    \end{align}
The difference of $\overline{\Rcal_{\new,\max}^{\Hcal_q}}(q^2,X,T)$ and $\overline{\Rcal^{\Xcal}_{\max}}(q^2,X,T)$ is 
    $$
      \overline{\Rcal_{\new,\max}^{\Hcal_q}}(q^2,X,T)-\overline{\Rcal^{\Xcal}_{\max}}(q^2,X,T)
=
        \frac{(q^3-5q^2+3-M)(q^2+M)-(q^2-M)(q^3-2q^2+2q+1)}{(q^3-2q^2+2q+1)(q^2+M)},
    $$
    where we use $M$ to denote $X+T$ for simplicity.
    Its numerator is 
    $$
    N(q^2,M){:=}(q^3{-}5q^2{+}3{-}M)(q^2{+}M){-}(q^2{-}M)(q^3{-}2q^2{+}2q{+}1)=-3q^4{+}2q^3M{-}2q^3{-}8q^2M{+}2q^2{+}2qM{-}M^2{+}4M.
    $$ 
    By \eqref{eq:202510311559} and \eqref{eq:202510311600}, to prove this proposition, it suffices to prove that 
    $N(q^2,M)>0$ for any $3q\leq M\leq q^2-2$ and $q\geq 8$.
    
    Note that $N(q^2,M)$ is a concave function of $M$ for any fixed $q$ since it is a quadratic polynomial of $M$ with leading coefficient $-1$.
    To prove this proposition, it suffices to prove that $N(q^2,M)|_{M=3q}=3q^4-26q^3-q^2+12q$ and $N(q^2,M)|_{M=(q^2-2)}=2q^5-12q^4-4q^3+26q^2-4q-12$ are both positive for any $q\geq 9$, which are straightforward to verify. The proof is complete.
\end{proof}

\subsection{Comparison with XSTPIR Scheme based on Elliptic curves and Hyperelliptic Curves}
\label{sec:4.hyper}
In this subsection, we make comparisons between the maximum rates of our XSTPIR schemes with existing XSTPIR schemes based on elliptic curves and hyperelliptic curves. According to Proposition~\ref{prop:hyperelliptic_maximum} and Remark~\ref{rem:elliptic_hyperelliptic_maximum}, it suffices to compare the maximum rates of our schemes with the upper bound $\overline{\Rcal^{g}_{\max}}(q,X,T)=\frac{2q-(X+T+8g+2)}{2q+X+T+4g+2}$ in \eqref{eq:makk_hyper_maximum_bound}.  

In the following, we make comparisons between the maximum rate of our XSTPIR schemes based on rational curves (Theorem~\ref{thm:our_XSTPIR_Rational}, Corollary~\ref{cor:Our_rational_maximum}, or see Row 5 of Table~\ref{tab:allmaximumrates}) and the upper bound of the maximum rate of the XSTPIR schemes based on hyperelliptic curves (Theorem~\ref{thm:makk_hyperelliptic}, Proposition~\ref{prop:hyperelliptic_maximum}, or see Row 3 of Table~\ref{tab:allmaximumrates}).
\begin{proposition}\label{prop:202511222144}
Assume that the conditions~\eqref{eq:condition_our_rational} and \eqref{eq:makk_hyper_condition} are satisfied, and $2\leq X+T\leq \floor{2\sqrt{3}q}-8$, $q\geq \undetermined$, $g\geq 1$. Then we have 
$$
\Rcal_{\new,\max}^{\Xcal}(q,X,T)> \overline{\Rcal^{g}_{\max}}(q,X,T).
$$
\end{proposition} 
\begin{proof}
By \eqref{eq:our_rational_maximum}, we have
    \begin{equation}\label{eq:202510291711}
        \Rcal_{\new,\max}^{\Xcal}(q,X,T)\geq \frac{2\parentheses{\frac{q-X-T-3}{2}-\frac{1}{2}}}{2\parentheses{\frac{q-X-T-3}{2}-\frac{1}{2}}+X+T+2}= \frac{q-X-T-4}{q-2}\xlongequal{\text{(denoted by)}}\overline{\Rcal_{\new,\max}^{\Xcal}}(q,X,T)
    \end{equation} 
Recall the upper bound $\overline{\Rcal^{g}_{\max}}(q,X,T)$ given in \eqref{eq:makk_hyper_maximum_bound}:
        \begin{align}\label{eq:202510291716}
           \overline{\Rcal^{g}_{\max}}(q,X,T)=\frac{2q-(X+T+8g+2)}{2q+X+T+4g+2} .
    \end{align} 
Consider the difference
\begin{align} 
\overline{\Rcal_{\new,\max}^{\Xcal}}(q,X,T)- \overline{\Rcal^{g}_{\max}}(q,X,T)
=\frac{(q-M-4)(2q+M+4g+2)-(2q-(M+8g+2))(q-2)}{(q-2)(2q+M+4g+2)}, \label{eq:202510301436}
\end{align}
where we use $M$ to denote $X+T$ for simplicity.
For the numerator of \eqref{eq:202510301436}, we have   
\begin{align*}
    (q-M-4)(2q+M+4g+2)-(2q-(M+8g+2))(q-2)&=12qg-4gM-M^2-32g-8M-12\\
    &=12qg-12-(M+4g)(M+8)\\
    (\text{when }M\geq 2\text{ and }g\geq 1)&\geq 12qg-(M+\max\{4g+2,8\})^2\\
    &\xlongequal{\text{(denoted by)}} N(q,g,M).
\end{align*}
By the above inequality and \eqref{eq:202510291711}, \eqref{eq:202510291716}, \eqref{eq:202510301436}, to prove this proposition, it suffices to show that $N(q,g,M)>0$ for any positive integers $q,g,M$ satisfying 
\begin{equation}\label{eq:202510301541}
    M\leq q-5 \text{ (i.e., condition \eqref{eq:condition_our_rational})}, M\leq 2q-8g-2 \text{ (i.e., condition \eqref{eq:makk_hyper_condition})}, 2\leq M\leq \floor{2\sqrt{3q}}-8, q\geq \undetermined, g\geq 1.
\end{equation}

Note that $N(q,g,M)>0$ for any $-\sqrt{12qg}-\max\{4g+2,8\}\leq M\leq \sqrt{12qg}-\max\{4g+2,8\}$ and $q\geq 1,g\geq 1$. To prove this proposition, it suffices to show that 
\begin{align}\label{eq:202510251631}
-\sqrt{12qg}-\max\{4g+2,8\}\leq 2\leq \floor{2\sqrt{3q}}-8\leq \sqrt{12qg}-\max\{4g+2,8\}
\end{align} 
for any $q,g$ satisfying $q\geq \undetermined,1\leq g\leq \frac{q-2}{4}$ \footnote{Here, we impose a stronger constraint $1\leq g\leq \frac{q-2}{4}$ rather than $g\geq 1$ as in \eqref{eq:202510301541}. This is because when $g>\frac{q-2}{4}$, no valid $M$ exists, as $M\leq 2q-8g-2$ and $2\leq M$ by \eqref{eq:202510301541}.}.
For this, we consider the following two cases according to the value of $g$:
\begin{itemize}
    \item[1)] $q\geq \undetermined,g=1$. In this case, \eqref{eq:202510251631} holds directly.
    \item[2)] $q\geq \undetermined,2\leq g\leq \frac{q-2}{4}$. In this case, to prove \eqref{eq:202510251631}, it suffices to prove that the function $\sqrt{12qg}-4g-2-(2\sqrt{3q}-8)$ is positive for any $q\geq \undetermined,2\leq g\leq \frac{q-2}{4}$. Since 
    $\pfrac{(\sqrt{12qg}-4g-2-(2\sqrt{3q}-8))}{g}=\frac{\sqrt{3q}}{\sqrt{g}}-4$. The function $\sqrt{12qg}-4g-2-(2\sqrt{3q}-8)$ of $g$ is increasing on $2\leq g\leq \frac{3}{16}q$ and decreasing on $\frac{3}{16}q\leq g\leq \frac{q-2}{4}$ for any fixed $q\geq 11$. Note that for any $q\geq \undetermined$, we have
    \begin{align*}
        &\sqrt{12qg}-4g-2-(2\sqrt{3q}-8)|_{g=2}=2\sqrt{6q}-2\sqrt{3q}-2> 0,\\
        &\sqrt{12qg}-4g-2-(2\sqrt{3q}-8)|_{g={\frac{q-2}{4}}}=\sqrt{3q(q-2)}-(q-2)-2\sqrt{3q}+6\geq \sqrt{3}(q-2)-2\sqrt{3q}+6>0.
    \end{align*}
\end{itemize}
The proof is complete.
\end{proof}
In the following, we make comparisons between the maximum rate of our XSTPIR schemes based on Hermitian curves (Theorem~\ref{thm:our_XSTPIR_Hermitian}, Corollary~\ref{cor:Our_Hermitian_maximum}, or see Row 6 of Table~\ref{tab:allmaximumrates}) and the upper bound of the maximum rate of the XSTPIR schemes based on hyperelliptic curves (Theorem~\ref{thm:makk_hyperelliptic}, Proposition~\ref{prop:hyperelliptic_maximum}, or see Row 3 of Table~\ref{tab:allmaximumrates}).
Note that in the following proposition, the field size is $q^2$, rather than $q$.
\begin{proposition}\label{prop:202511222143} 
 When $q\geq 7$, the condition~\eqref{eq:condition_our_hermitian} is weaker than the condition~\eqref{eq:makk_hyper_condition} \footnote{Note that condition~\eqref{eq:makk_hyper_condition} now becomes $\overline{\Rcal^{g}_{\max}}(q^2,X,T)=\frac{2q^2-(X+T+8g+2)}{2q^2+X+T+4g+2}\geq 0$ since the field size here is $q^2$.} for any $g\geq 1$. Moreover, when both conditions~\eqref{eq:condition_our_hermitian} and~\eqref{eq:makk_hyper_condition} are satisfied (as long as \eqref{eq:makk_hyper_condition} is satisfied), and one of the following extra conditions
 
    1) $X+T\geq 4q$, $q\geq 14, g=1$,
    
    2) $X+T\geq 4q$, $q\geq 11 , g\geq 2$,
    
    3) $X+T\geq \floor{3.4q}$, $q\geq 28, g\geq 1$
is satisfied, we have $$\Rcal_{\new,\max}^{\Hcal_q}(q^2,X,T)> \overline{\Rcal^{g}_{\max}}(q^2,X,T).$$
\end{proposition} 
\begin{proof} 
Assume that $q\geq 7$ and the condition~\eqref{eq:makk_hyper_condition} is satisfied, that is, $\overline{\Rcal^{g}_{\max}}(q^2,X,T)=\frac{2q^2-(X+T+8g+2)}{2q^2+(X+T+4g+2)}\geq 0$. Then we have $X+T\leq 2q^2$, which implies $\floor{\frac{q^3-3q^2-3q+2-X-T}{2q}}\geq \frac{q^3-3q^2-3q+2-2q^2}{2q}-1\geq \frac{q-1}{2}$. This means that the condition~\eqref{eq:condition_our_hermitian} is weaker than the condition~\eqref{eq:makk_hyper_condition}. 
 
We now make a comparison between $\Rcal_{\new,\max}^{\Hcal_q}(q^2,X,T)$ and $\overline{\Rcal^{g}_{\max}}(q^2,X,T)$.
Before proceeding, we recall a lower bound of $\Rcal_{\new,\max}^{\Hcal_q}(q^2,X,T)$ in \eqref{eq:202510311559} for simplicity of comparison: 
    \begin{align}  \label{eq:202510291603}
           \Rcal_{\new,\max}^{\Hcal_q}(q^2,X,T) \geq \frac{q^3-5q^2+3-X-T}{q^3-2q^2+2q+1}
           \xlongequal{\text{(denoted by)}}\overline{\Rcal_{\new,\max}^{\Hcal_q}}(q^2,X,T) 
    \end{align} 
By replacing $q$ with $q^2$ in \eqref{eq:makk_hyper_maximum_bound}, we obtain
    \begin{align}\label{eq:202510291615}
           \overline{\Rcal^{g}_{\max}}(q^2,X,T)=\frac{2q^2-(X+T+8g+2)}{2q^2+X+T+4g+2}.
    \end{align} 
Their difference is
\begin{align}\notag
   &\overline{\Rcal_{\new,\max}^{\Hcal_q}}(q^2,X,T)-\overline{\Rcal^{g}_{\max}(q^2,X,T)}\\\label{eq:202601080113}
&=\frac{-6q^4+12q^3g+2q^3M-36q^2g-9q^2M-10q^2+16qg+2qM-4gM-M^2+4q+20g+2M+8}{(q^3-2q^2+2q+1)(2q^2+X+T+4g+2)},
\end{align}
where we set $M=X+T$ for simplicity. By \eqref{eq:202510291603}, \eqref{eq:202510291615}, and \eqref{eq:202601080113}, to prove this proposition, it suffices to show
$$
N(q,g,M):=-6q^4+12q^3g+2q^3M-36q^2g-9q^2M-10q^2+16qg+2qM-4gM-M^2+4q+20g+2M+8>0
$$
for any triples $q,g,M=X+T$ satisfying the conditions of this proposition.

Note that $N(q,g,M)$ is a concave function of $M$ for any fixed $q$ and $g$ since it is a quadratic polynomial of $M$ with leading coefficient $-1$. Furthermore, observe that condition \eqref{eq:makk_hyper_condition} implies $M=X+T\leq 2q^2-8g-2\leq 2q^2$. To prove this proposition, it suffices to show that $N(q,g,M)|_{M=4q}>0$ and $N(q,g,M)|_{M=2q^2}>0$ for any integers $q,g$ satisfying
\begin{itemize}
    \item[1)] $q\geq 14$ and $g=1$, or
    \item[2)] $q\geq 11$ and $g>1$;
\end{itemize}
and show that $N(q,g,M)|_{M=3.4q-1}>0$ and $N(q,g,M)|_{M=2q^2}>0$ for any integers $q,g$ satisfying
\begin{itemize}
   \item[3)] $q\geq 28$ and $g\geq 1$.
\end{itemize}
Indeed, by direct computation, we have
\begin{align*}
    &N(q,g,M)|_{M=4q}=2q^4+12q^3g-36q^3-36q^2g-18q^2+12q+20g+8,\\
    &N(q,g,M)|_{M=3.4q-1}=0.8q^4+12q^3g-32.6q^3-36q^2g-5.76q^2+2.4qg+15.6q+24g+5,\\
    &N(q,g,M)|_{M=2q^2}=4q^5-28q^4+12q^3g+4q^3-44q^2g-6q^2+16qg+4q+20g+8.
\end{align*}
They are all positive under the corresponding conditions of $q,g$ described in 1), 2), 3) above. The proof is complete.
\end{proof}

\subsection{Comparison with Existing XSTPIR Schemes based on Hermitian Curves}
\label{sec:4.Hermitian}
In this subsection, we compare the maximum rate of our proposed Hermitian-curve-based XSTPIR scheme (Theorem~\ref{thm:our_XSTPIR_Hermitian} and Corollary~\ref{cor:Our_Hermitian_maximum}, or see Row 6 of Table~\ref{tab:allmaximumrates}) with the existing Hermitian-curve-based XSTPIR scheme proposed in \cite{ghiandoni2025agcodes} (Theorem~\ref{thm:GGMT_Hermitian} and Proposition~\ref{prop:GGMT_hermitian_maximum}, or see Row 4 of Table~\ref{tab:allmaximumrates}).
\begin{proposition}\label{prop:202510191639}
    For any prime power $q\geq 5$ \footnote{When $q\leq 4$, neither the condition \eqref{eq:condition_our_hermitian} nor the condition \eqref{eq:ghi_hermitian_condition} can be satisfied for any $X\geq 1,T\geq 1$.}, the condition~\eqref{eq:condition_our_hermitian} is weaker than the condition~\eqref{eq:ghi_hermitian_condition}. Moreover, when $q\geq 6$ and conditions \eqref{eq:condition_our_hermitian} and \eqref{eq:ghi_hermitian_condition} are satisfied (as long as the condition~\eqref{eq:ghi_hermitian_condition} is satisfied), we have
    $$\Rcal_{\new,\max}^{\Hcal_q}(q^2,X,T)> \Rcal^{\Hcal_q}_{\max}(q^2,X,T).$$
\end{proposition}
\begin{proof}
Let us show that the condition~\eqref{eq:ghi_hermitian_condition} implies the condition~\eqref{eq:condition_our_hermitian}. 

Since $\frac{q^3-3q^2+q+1-X-T}{2q}-\frac{q^3-3q^2-3q+2-X-T}{2q}=2-\frac{1}{2q}\leq 2$, we have 
$$\floor{\frac{q^3-3q^2+q+1-X-T}{2q}}-\floor{\frac{q^3-3q^2-3q+2-X-T}{2q}}\leq 2.$$
Since also $q-1-\frac{q-1}{2}=\frac{q-1}{2}\geq 2$ for any $q\geq 5$, the condition~\eqref{eq:ghi_hermitian_condition} implies \eqref{eq:condition_our_hermitian}.

Next, we conduct a comparison between $\Rcal_{\new,\max}^{\Hcal_q}(q^2,X,T)$ and $\Rcal^{\Hcal_q}_{\max}(q^2,X,T)$ when $q\geq 6$ and the condition~\eqref{eq:ghi_hermitian_condition} is satisfied.
Let $X+T=2qt+b$ for two integers $t,b$, where $-(q-2)\leq b\leq q+1$.
We divide our discussion into two cases: 1) $b=-(q-2)$ and 2) $-(q-3)\leq b\leq q+1$.
\begin{itemize}
    \item[1)] $b=-(q-2)$. Since $\frac{q^3-3q^2}{2q}=\frac{q(q-1)}{2}-q \in \Zbb$, we have $\floor{\frac{q^3-3q^2-3q+2-X-T}{2q}}=\frac{q(q-1)}{2}-q-t-1$ and $\floor{\frac{q^3-3q^2+q+1-X-T}{2q}}=\frac{q(q-1)}{2}-q-t$. Applying them to \eqref{eq:our_hermitian_maximum} and \eqref{eq:ghi_hermitian_maximum}, the difference of two maximum rates for comparison is
    \begin{align}
    \Rcal_{\new,\max}^{\Hcal_q}(q^2,X,T)-\Rcal^{\Hcal_q}_{\max}(q^2,X,T)&=\frac{2q(\frac{q(q-1)}{2}-q-t-1)-\frac{q(q-1)}{2}}{2q(\frac{q(q-1)}{2}-q-t-1)-\frac{q(q-1)}{2}+2qt+b+3q^2+2q-2} \notag \\\notag
    &\qquad\qquad\qquad-\frac{q(\frac{q(q-1)}{2}-q-t)-\frac{q(q-1)}{2}}{q(\frac{q(q-1)}{2}-q-t)-\frac{q(q-1)}{2}+2qt+b+3q^2-q-2}\\\notag
    &\geq\frac{2q(\frac{q(q-1)}{2}-q-t-2)-\frac{q(q-1)}{2}}{2q(\frac{q(q-1)}{2}-q-t-2)-\frac{q(q-1)}{2}+2qt+b+3q^2+2q-2}\\ \notag 
    &\qquad\qquad\qquad-\frac{q(\frac{q(q-1)}{2}-q-t)-\frac{q(q-1)}{2}}{q(\frac{q(q-1)}{2}-q-t)-\frac{q(q-1)}{2}+2qt+b+3q^2-q-2}\\ \label{eq:202510211814}
    &\xlongequal{\text{(denoted by)}}D(q,t,b).
    \end{align}

    \item[2)] $-(q-3)\leq b\leq q+1$. Since $\frac{q^3-3q^2}{2q}=\frac{q(q-1)}{2}-q \in \Zbb$, we have $\floor{\frac{q^3-3q^2-3q+2-X-T}{2q}}=\frac{q(q-1)}{2}-q-t-2$ and $\floor{\frac{q^3-3q^2+q+1-X-T}{2q}}=\frac{q(q-1)}{2}-q-t$.
    Applying them to \eqref{eq:our_hermitian_maximum} and \eqref{eq:ghi_hermitian_maximum}, the difference of two maximum rates for comparison is
   \begin{align}\notag 
    \Rcal_{\new,\max}^{\Hcal_q}(q^2,X,T)-\Rcal^{\Hcal_q}_{\max}(q^2,X,T)&=
    \frac{2q(\frac{q(q-1)}{2}-q-t-2)-\frac{q(q-1)}{2}}{2q(\frac{q(q-1)}{2}-q-t-2)
    -\frac{q(q-1)}{2}+2qt+b+3q^2+2q-2}\\ \notag
    &\qquad\qquad\qquad-\frac{q(\frac{q(q-1)}{2}-q-t)-\frac{q(q-1)}{2}}{q(\frac{q(q-1)}{2}-q-t)-\frac{q(q-1)}{2}+2qt+b+3q^2-q-2}\\
    &\xlongequal{\text{(see \eqref{eq:202510211814})}}D(q,t,b). \label{eq:202512171558}
    \end{align}
\end{itemize}
Recall that the condition~\eqref{eq:ghi_hermitian_condition} must be satisfied. That is, $\floor{\frac{q^3-3q^2+q+1-X-T}{2q}}=\frac{q(q-1)}{2}-q-t\geq q-1$, which implies $t\leq \frac{q^2}{2}-\frac{5}{2}q+1$. By \eqref{eq:202510211814} and \eqref{eq:202512171558}, to prove this proposition, it suffices to prove that $D(q,t,b)>0$ for any integers $q,t,b$ satisfying $q\geq 6$, $0\leq t\leq \frac{q^2}{2}-\frac{5}{2}q+1$, $-(q-2)\leq b\leq q+1$, and $2qt+b\geq 2$. Next, we proceed to prove this.

By direct computation, we have 
\begin{align}\label{eq:202512171619}
    D(q,t,b){=}
    \frac{\frac{3}{2}q^5+q^4t-\frac{13}{2}q^4-6q^3t-2q^2t^2+\frac{1}{2}q^3b-\frac{11}{2}q^3-4q^2t-\frac{3}{2}q^2b-qtb+\frac{11}{2}q^2+2qt-4qb+8q}{\parentheses{2q(\frac{q(q-1)}{2}{-}q{-}t{-}2)
    {-}\frac{q(q-1)}{2}{+}2qt{+}b{+}3q^2{+}2q{-}2}\parentheses{q(\frac{q(q-1)}{2}{-}q{-}t)-\frac{q(q-1)}{2}{+}2qt{+}b{+}3q^2{-}q{-}2}}.
\end{align}
Note that its denominator is positive for those $q,t,b$. Hence, to prove that $D(q,t,b)>0$ for any integers $q,t,b$ satisfying $q\geq 6$, $0\leq t\leq \frac{q^2}{2}-\frac{5}{2}q+1$, $-(q-2)\leq b\leq q+1$, and $2qt+b\geq 2$, we only need to prove that its numerator 
\begin{align}\label{eq:202510221425}
    N(q,t,b)=\frac{3}{2}q^5+q^4t-\frac{13}{2}q^4-6q^3t-2q^2t^2+\frac{1}{2}q^3b-\frac{11}{2}q^3-4q^2t-\frac{3}{2}q^2b-qtb+\frac{11}{2}q^2+2qt-4qb+8q
\end{align} 
is positive. 
Note that $\pfrac{N(q,t,b)}{b}=\frac{1}{2}q^3-\frac{3}{2}q^2-qt-4q > 0$ when $q\geq 6$ and $t\leq \frac{1}{2}q^2-\frac{5}{2}q+1,$ in which case we have $N(q,t,b)\geq N(q,t,-(q-2))$ for any $-(q-2)\leq b\leq q+1$. Therefore, 
to prove that $N(q,t,b)>0$ for any integers $q,t,b$ satisfying $q\geq 6$, $0\leq t\leq \frac{q^2}{2}-\frac{5}{2}q+1$, $-(q-2)\leq b\leq q+1$, and $2qt+b\geq 2$, it suffices to prove that 
$$
N(q,t,b)|_{b=-(q-2)}=\frac{3}{2}q^5+q^4t-7q^4-6q^3t-2q^2t^2-3q^3-3q^2t+\frac{13}{2}q^2
$$
is positive for any integers $q,t$ satisfying $q\geq 6$, $0\leq t\leq \frac{q^2}{2}-\frac{5}{2}q+1$. This is true because, for any fixed $q\geq 6$,
\begin{itemize}
    \item[1)] $N(q,t,b)|_{b=-(q-2)}$ is a concave function of $t$ since it is a quadratic polynomial of $t$ with leading coefficient $-2q^2$,
    \item[2)] $N(q,t,b)|_{b=-(q-2), t=0}=\frac{3}{2}q^5 - 7q^4-3q^3+\frac{13}{2}q^2>0$ and $N(q,t,b)|_{b=-(q-2), t=\frac{q^2}{2}-\frac{5}{2}q+1}=q^5-7q^4+\frac{17}{2}q^3+\frac{3}{2}q^2>0$.
\end{itemize}
The proof is complete.
\end{proof}

\begin{remark}
  From the above comparison, we cannot obtain a quantitative sense of the improvement. We briefly estimate how much larger our maximum rate can be than that of the existing scheme based on the Hermitian function fields when $X+T=\frac{2}{5}q^3$. For this, we only need to recall the lower estimation $D(q,t,b)$ of the difference $\Rcal_{\new,\max}^{\Hcal_q}(q^2,X,T)-\Rcal^{\Hcal_q}_{\max}(q^2,X,T)$, in \eqref{eq:202510211814} and \eqref{eq:202512171558}. The value  $D(q,t,b)|_{t=\frac{q^2}{5},b=0}$ is a lower estimation of the difference $\Rcal_{\new,\max}^{\Hcal_q}(q^2,X,T)-\Rcal^{\Hcal_q}_{\max}(q^2,X,T)$ when $X+T=\frac{2q^3}{5}$ (recall that we set $X+T=2qt+b$ in \eqref{eq:202510211814} and \eqref{eq:202512171558}). Using the explicit expression \eqref{eq:202512171619} of $D(q,t,b)$, a direct computation yields that 
  $$D(q,t,b)|_{t=\frac{q^2}{5},b=0}=\frac{\frac{3}{25}q^6 + \frac{3}{10}q^5 - \frac{73}{10}q^4 - \frac{51}{10}q^3 + \frac{11}{2}q^2 + 8q}{\frac{7}{10}q^6 + \frac{13}{20}q^5 - \frac{41}{20}q^4 - \frac{93}{20}q^3 - \frac{1}{4}q^2 + 4q + 4},$$ which tends to $\frac{\frac{3}{25}}{\frac{7}{10}}\approx 0.1714$ when $q\rightarrow \infty$. And it holds $D(q,t,b)|_{t=\frac{q^2}{5},b=0}\geq 0.14$ for any $q\geq 15$. That is to say, when $q\geq 15$ and $X+T=\frac{2q^2}{5}$, $\Rcal_{\new,\max}^{\Hcal_q}(q^2,X,T)$ is at least 0.14 larger than $\Rcal_{\max}^{\Hcal_q}(q^2,X,T)$. This represents a significant improvement.
\end{remark}

\subsection{Overall Comparisons}
In this subsection, we summarize those comparisons in Section~\ref{sec:4.rational}, \ref{sec:4.hyper}, \ref{sec:4.Hermitian}, to get overall comparisons between the maximum rates of our proposed XSTPIR schemes and all existing XSTPIR schemes.

When the fixed field size is $q^2\geq 14^2$ and $X+T\geq 4q$, our XSTPIR schemes based on Hermitian curves have higher maximum rates than all existing XSTPIR schemes. 
\begin{corollary}\label{cor:comparison_our_hermitian}
For any $q\geq 14$ and $X,T$ satisfying $4q\leq X+T$ and $\floor{\frac{q^3-3q^2-3q+2-X-T}{2q}}\geq \frac{q-1}{2}$, we have 
$$
\Rcal^{\Hcal_q}_{\new,\max}(q^2,X,T)>\max\braces{\Rcal^{\Hcal_q}_{\max}(q^2,X,T),\max_{g\geq 1}\overline{\Rcal^{g}_{\max}}(q^2,X,T),\Rcal^{\Xcal}_{\max}(q^2,X,T)}.
$$
(For the terms in this inequality, see \eqref{eq:our_hermitian_maximum}, \eqref{eq:ghi_hermitian_maximum}, \eqref{eq:makk_hyper_maximum_bound}, \eqref{eq:jia_makk_rational_maximum}. When $X,T$ do not satisfy their corresponding conditions~\eqref{eq:condition_our_hermitian}, \eqref{eq:ghi_hermitian_condition}, \eqref{eq:makk_hyper_condition}, \eqref{eq:jia_makk_rational_condition}, the value will be specified to be $0$ since in this case no valid scheme exists.)
\end{corollary}
    
\begin{proof}  
We prove this corollary by proving the following three separate claims for any prime power $q\geq 14$ and any positive integers $X,T$ satisfying $4q\leq X+T$ and $\floor{\frac{q^3-3q^2-3q+2-X-T}{2q}}\geq \frac{q-1}{2}$:  
\begin{itemize}
    \item[1)] $\Rcal^{\Hcal_q}_{\new,\max}(q^2,X,T)> \Rcal^{\Xcal}_{\max}(q^2,X,T)$. It is true by  Proposition~\ref{prop:202511222203}. 
    \item[2)] $\Rcal^{\Hcal_q}_{\new,\max}(q^2,X,T)> \max_{g\geq 1}\overline{\Rcal^{g}_{\max}}(q^2,X,T)$. It is true by Proposition~\ref{prop:202511222143} 1), 2).
    \item[3)] $\Rcal^{\Hcal_q}_{\new,\max}(q^2,X,T)> \Rcal^{\Hcal_q}_{\max}(q^2,X,T)$. It is true by Proposition~\ref{prop:202510191639}.
\end{itemize} 
\end{proof}

When the fixed field size is $q^2\geq 28^2$ and $X+T\geq 4$, our two XSTPIR schemes jointly provide the largest known maximum rates.
\begin{corollary}\label{cor:comparison_our_rational_and_hermitian}
For any $q\geq 28$ and $X,T$ satisfying $4\leq X+T$ and $\floor{\frac{q^3-3q^2-3q+2-X-T}{2q}}\geq \frac{q-1}{2}$, we have 
$$
\max\braces{\Rcal^{\Hcal_q}_{\new,\max}(q^2,X,T),\Rcal^{\Xcal}_{\new,\max}(q^2,X,T)}>\max\braces{\Rcal^{\Hcal_q}_{\max}(q^2,X,T),\max_{g\geq 1}\overline{\Rcal^{g}_{\max}}(q^2,X,T), \Rcal^{\Xcal}_{\max}(q^2,X,T)}.
$$
(For the terms in this inequality, see \eqref{eq:our_hermitian_maximum}, \eqref{eq:our_rational_maximum}, \eqref{eq:ghi_hermitian_maximum}, \eqref{eq:makk_hyper_maximum_bound}, \eqref{eq:jia_makk_rational_maximum}. When $X,T$ do not satisfy their corresponding conditions \eqref{eq:condition_our_hermitian}, \eqref{eq:condition_our_rational}, \eqref{eq:ghi_hermitian_condition}, \eqref{eq:makk_hyper_condition}, \eqref{eq:jia_makk_rational_condition}, the value will be specified to be $0$ since in this case no valid scheme exists.)
\end{corollary}
    
\begin{proof}  
We prove this corollary by proving the following three separate claims for any prime power $q\geq 28$ and any positive integers $X,T$ satisfying $4\leq X+T$ and $\floor{\frac{q^3-3q^2-3q+2-X-T}{2q}}\geq \frac{q-1}{2}$: 
\begin{itemize}
    \item[1)] $\max\braces{\Rcal^{\Hcal_q}_{\new,\max}(q^2,X,T),\Rcal^{\Xcal}_{\new,\max}(q^2,X,T)}> \Rcal^{\Xcal}_{\max}(q^2,X,T)$. It is true by Propositions~\ref{prop:202511222202} and \ref{prop:202511222203}. 
    
    \item[2)] $\max\braces{\Rcal^{\Hcal_q}_{\new,\max}(q^2,X,T),\Rcal^{\Xcal}_{\new,\max}(q^2,X,T)}> \max_{g\geq 1}\overline{\Rcal^{g}_{\max}}(q^2,X,T)$. It is true by Proposition~\ref{prop:202511222144} and Proposition~\ref{prop:202511222143} 3).
    \item[3)] $\max\braces{\Rcal^{\Hcal_q}_{\new,\max}(q^2,X,T),\Rcal^{\Xcal}_{\new,\max}(q^2,X,T)}> \Rcal^{\Hcal_q}_{\max}(q^2,X,T)$. It is true by Proposition~\ref{prop:202510191639}.
\end{itemize} 
\end{proof}
The following Figure~\ref{fig:comparison} visualizes the comparison of maximum rates given in Corollary~\ref{cor:comparison_our_rational_and_hermitian}, when the field size is equal to $q^2=29^2$. 
To reduce the complexity of the figure and avoid excessive line overlapping, we only plot $\overline{\Rcal^{g}_{\max}}(q^2,X,T)$ for the case $g=1$.  
It can be seen from the figure that at least one of the maximum rates of our two XSTPIR schemes exceeds all the maximum rates of existing XSTPIR schemes, for any $X\geq 2$ (we set $X=T$ in Figure~\ref{fig:comparison} for simplicity).
\begin{figure}[H]
    \centering
    \includegraphics[width=1\linewidth]{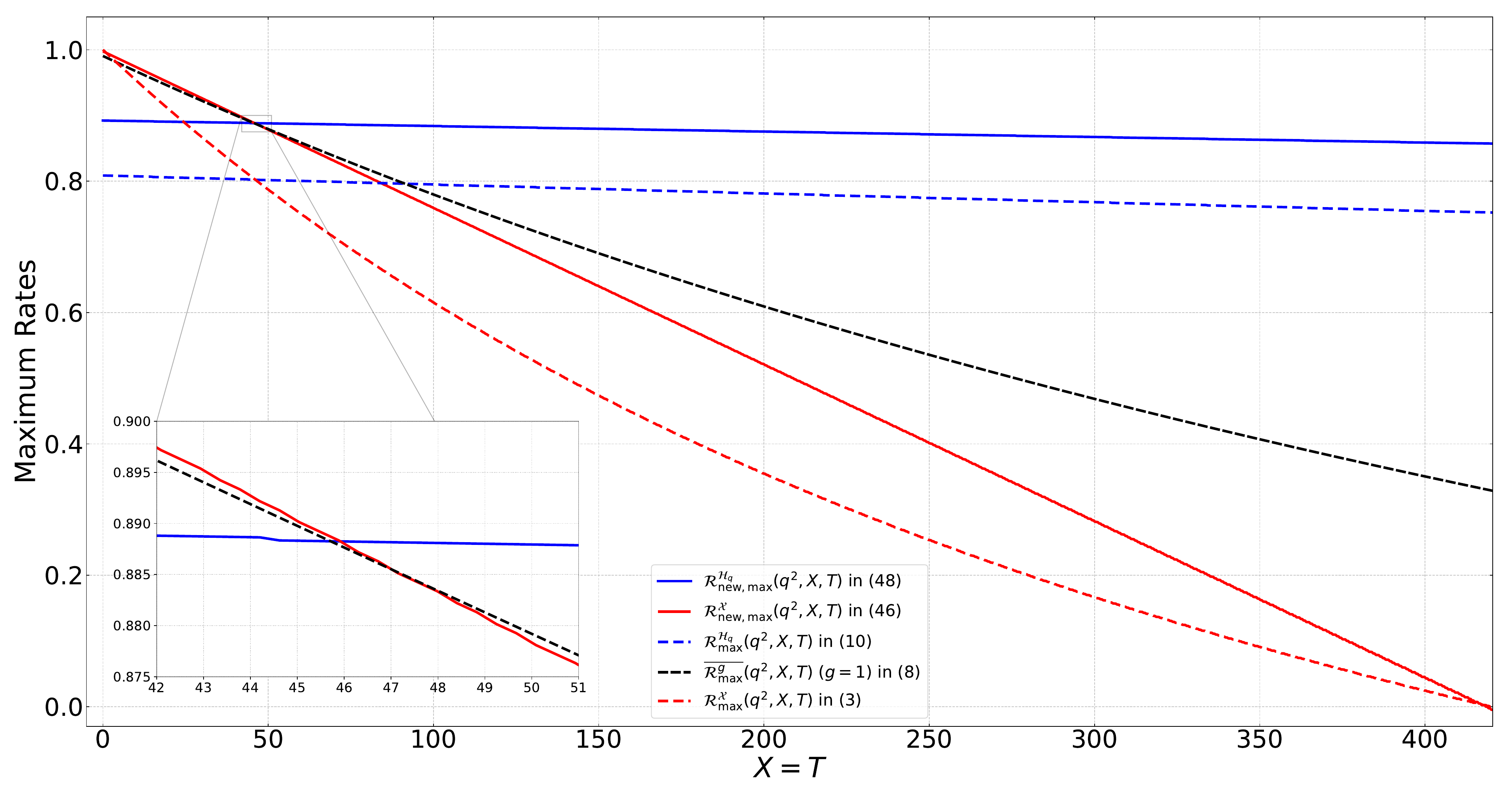}
    \caption{Comparison of maximum rates for $q=29$ (field size $q^2=841$).}
    \label{fig:comparison}
\end{figure}
\section{Conclusions}\label{sec:5} 
    In this paper, we focus on constructing new XSTPIR schemes with higher maximum rates. Moving beyond the mainstream approach of seeking curves with higher genus and more rational points, we achieve this goal by enhancing the utilization efficiency of rational points on curves that have already been considered in previous work. Specifically, by introducing a family of bases of $\Span_{\F_q}\{1,x\dots,x^{k-1}\}$ as an alternative to the Lagrange interpolation basis, we demonstrate that significant rate improvements are possible. We believe our approach provides valuable insights for future research, even when applied to other types of algebraic curves, such as Norm-Trace curves, Ree curves, and Suzuki curves.
\bibliographystyle{IEEEtran}
\bibliography{References.bib}

@book{stichtenoth2009algebraic,
  title={Algebraic Function Fields and Codes},
  author={Stichtenoth, Henning},
  series={Graduate Texts in Mathematics, vol. 254},
  address={Berlin, Germany},
  volume={254},
  year={2009},  
  publisher={Springer-Verlag},
}

@book{niederreiter2001rational,
author = {Niederreiter, Harald and Xing, Chaoping},
title = {Rational Points on Curves over Finite Fields: Theory and Applications},
year = {2001},
isbn = {0521665434},
series={London Mathematical Society Lecture Note Series, vol. 285},
publisher = {Cambridge University Press},
address = {USA}
}

@INPROCEEDINGS{makkonen2024algebraic,
  author={Makkonen, Okko and Karpuk, David A. and Hollanti, Camilla},
  booktitle={Proc. IEEE Int. Symp. Inf. Theory (ISIT)}, 
  title={Algebraic Geometry Codes for Cross-Subspace Alignment in Private Information Retrieval}, 
  year={2024},
  volume={},
  number={},
  pages={2874-2879},
  doi={10.1109/ISIT57864.2024.10619560}}

@misc{ghiandoni2025agcodes,
       title={{AG codes from the Hermitian curve for Cross-Subspace Alignment in Private Information Retrieval}}, 
      author={Francesco Ghiandoni and Massimo Giulietti and Enrico Mezzano and Marco Timpanella},
      year={2025},
      eprint={2508.19459},
      archivePrefix={arXiv},
      primaryClass={math.AG},
      url={https://arxiv.org/abs/2508.19459}, 
}

@ARTICLE{jia2019cross,
  author={Jia, Zhuqing and Sun, Hua and Jafar, Syed Ali},
  journal={IEEE Trans. Inf. Theory}, 
  title={Cross Subspace Alignment and the Asymptotic Capacity of  {$X$}-Secure {$T$}-Private Information Retrieval}, 
  year={2019},
  volume={65},
  number={9},
  pages={5783-5798},
  doi={10.1109/TIT.2019.2916079}}

@misc{makkonen2024secretsharingsecureprivate,
      title={Secret Sharing for Secure and Private Information Retrieval: A Construction Using Algebraic Geometry Codes}, 
      author={Okko Makkonen and David Karpuk and Camilla Hollanti},
      year={2024},
      eprint={2408.00542},
      archivePrefix={arXiv},
      primaryClass={cs.IT},
      url={https://arxiv.org/abs/2408.00542}, 
}

@book{niederreiter2009algebraic,
  title={Algebraic geometry in coding theory and cryptography},
  author={Niederreiter, Harald and Xing, Chaoping},
  year={2009},
  address={NJ, USA},
  publisher={Princeton University Press}
}

@inproceedings{CKGS95FOCS:PIR,
  title={Private information retrieval},
  author={Chor, B and Goldreich, O and Kushilevitz, E and Sudan, M},
  booktitle={Proc. 36th IEEE Symp. Found. Comput. Sci., Milwaukee, WI, USA},
  pages={41--50},
  year={Oct. 1995},
}

@ARTICLE{Sun&Jafar16:CapacityPIR,
  author={Sun, Hua and Jafar, Syed Ali},
  journal={IEEE Trans. Inf. Theory}, 
  title={The Capacity of Private Information Retrieval}, 
  year={2017},
  volume={63},
  number={7},
  pages={4075-4088},
  doi={10.1109/TIT.2017.2689028}
}

@ARTICLE{Sun&Jafar16:ColludPIR,
  author={Sun, Hua and Jafar, Syed Ali},
  journal={IEEE Trans. Inf. Theory}, 
  title={The Capacity of Robust Private Information Retrieval With Colluding Databases}, 
  year={2018},
  volume={64},
  number={4},
  pages={2361-2370},
  doi={10.1109/TIT.2017.2777490}
}

@ARTICLE{YLW20:CapaPIRCollud,
  author={Yao, Xinyu and Liu, Nan and Kang, Wei},
  journal={IEEE Trans. Inf. Theory}, 
  title={The Capacity of Private Information Retrieval Under Arbitrary Collusion Patterns for Replicated Databases}, 
  year={2021},
  volume={67},
  number={10},
  pages={6841-6855},
  doi={10.1109/TIT.2021.3100476}}

@ARTICLE{UAGJ22:SuPIR,
  author={Ulukus, Sennur and Avestimehr, Salman and Gastpar, Michael and Jafar, Syed A. and Tandon, Ravi and Tian, Chao},
  journal={IEEE J. Sel. Areas Commun.}, 
  title={Private Retrieval, Computing, and Learning: Recent Progress and Future Challenges}, 
  year={2022},
  volume={40},
  number={3},
  pages={729-748},
  doi={10.1109/JSAC.2022.3142358}}

@ARTICLE{jia2020x,
  author={Jia, Zhuqing and Jafar, Syed Ali},
  journal={IEEE Trans. Inf. Theory}, 
  title={{$X$}-Secure {$T$}-Private Information Retrieval From {MDS} Coded Storage With Byzantine and Unresponsive Servers}, 
  year={2020},
  volume={66},
  number={12},
  pages={7427-7438}}

@ARTICLE{banawan2019capacity,
  author={Banawan, Karim and Ulukus, Sennur},
  journal={IEEE Trans. Inf. Theory}, 
  title={The Capacity of Private Information Retrieval from Byzantine and Colluding Databases}, 
  year={2019},
  volume={65},
  number={2},
  pages={1206-1219},
  doi={10.1109/TIT.2018.2869154}}

@ARTICLE{sun2019capacitySym,
  author={Sun, Hua and Jafar, Syed Ali},
  journal={IEEE Trans. Inf. Theory}, 
  title={The Capacity of Symmetric Private Information Retrieval}, 
  year={2019},
  volume={65},
  number={1},
  pages={322-329},
  doi={10.1109/TIT.2018.2848977}}

@ARTICLE{wang2019PIRSym,
  author={Wang, Qiwen and Skoglund, Mikael},
  journal={IEEE Trans. Inf. Theory}, 
  title={On {PIR} and Symmetric {PIR} From Colluding Databases With Adversaries and Eavesdroppers}, 
  year={2019},
  volume={65},
  number={5},
  pages={3183-3197},
  doi={10.1109/TIT.2018.2878034}}

@ARTICLE{Banawan2018multi,
  author={Banawan, Karim and Ulukus, Sennur},
  journal={IEEE Trans. Inf. Theory}, 
  title={Multi-Message Private Information Retrieval: Capacity Results and Near-Optimal Schemes}, 
  year={2018},
  volume={64},
  number={10},
  pages={6842-6862},
  doi={10.1109/TIT.2018.2828310}}

@article{zhang2019general,
  title={A general private information retrieval scheme for {MDS} coded databases with colluding servers},
  author={Zhang, Yiwei and Ge, Gennian},
  journal={Des. Codes Cryptogr.},
  volume={87},
  number={11},
  pages={2611--2623},
  year={2019},
  publisher={Springer}
}

@article{hollanti2017private,
author = {Freij-Hollanti, Ragnar and Gnilke, Oliver W. and Hollanti, Camilla and Karpuk, David A.},
title = {Private Information Retrieval from Coded Databases with Colluding Servers},
journal = {SIAM J. Appl. Algebra Geom.},
volume = {1},
number = {1},
pages = {647-664},
year = {2017},
doi = {10.1137/16M1102562}, 
eprint = {https://doi.org/10.1137/16M1102562},
}

@INPROCEEDINGS{lin2018MDSPIR,
  author={Lin, Hsuan-Yin and Kumar, Siddhartha and Rosnes, Eirik and Graell i Amat, Alexandre},
  booktitle={Proc. IEEE Int. Symp. Inf. Theory (ISIT)}, 
  title={An {MDS-PIR} Capacity-Achieving Protocol for Distributed Storage Using Non-{MDS} Linear Codes}, 
  year={2018},
  volume={},
  number={},
  pages={966-970},
  doi={10.1109/ISIT.2018.8437804}}

@article{xu2018sub,
  title={On sub-packetization and access number of capacity-achieving {PIR} schemes for {MDS} coded non-colluding servers},
  author={Xu, Jingke and Zhang, Zhifang},
  journal={Sci. China Inf. Sci.},
  volume={61},
  number={10},
  pages={100306},
  year={2018},
  publisher={Springer}
}

@ARTICLE{hollanti2019tprivate,
  author={Freij-Hollanti, Ragnar and Gnilke, Oliver W. and Hollanti, Camilla and Horlemann-Trautmann, Anna-Lena and Karpuk, David and Kubjas, Ivo},
  journal={IEEE Trans. Inf. Theory}, 
  title={$t$-private information retrieval schemes using transitive codes}, 
  year={2019},
  volume={65},
  number={4},
  pages={2107-2118},
  doi={10.1109/TIT.2018.2871050}}

@INPROCEEDINGS{sunrui2025MDSTPIR,
  author={Sun, Rui and Zhang, Yiwei},
  booktitle={Proc. IEEE Int. Symp. Inf. Theory (ISIT)}, 
  title={{MDS-TPIR} Schemes: Disguise and Squeeze}, 
  year={2025},
  volume={},
  number={},
  pages={1-6},
  doi={10.1109/ISIT63088.2025.11195522}}

@ARTICLE{xu2025explicit,
  author={Xu, Jingke and Fang, Weijun},
  journal={IEEE Trans. Inf. Theory}, 
  title={Explicit Constructions of Capacity-Achieving {T-PIR} Schemes Over Small Fields via Generalized Minor Matrices}, 
  year={2025},
  volume={71},
  number={7},
  pages={5109-5129}, 
  doi={10.1109/TIT.2025.3565288}}

@INPROCEEDINGS{xu2018building,
  author={Xu, Jingke and Zhang, Zhifang},
  booktitle={Proc. IEEE Int. Symp. Inf. Theory (ISIT)}, 
  title={Building Capacity-Achieving {PIR} Schemes with Optimal Sub-Packetization over Small Fields}, 
  year={2018},
  volume={},
  number={},
  pages={1749-1753}, 
  doi={10.1109/ISIT.2018.8437880}}

@INPROCEEDINGS{xu2019capacityachiev,
  author={Xu, Jingke and Zhang, Yaqian and Zhang, Zhifang},
  booktitle={Proc. IEEE Int. Symp. Inf. Theory (ISIT)}, 
  title={A Capacity-Achieving {T-PIR} Scheme Based On {MDS} Array Codes}, 
  year={2019},
  volume={},
  number={},
  pages={1047-1051},
  doi={10.1109/ISIT.2019.8849312}}

@ARTICLE{xu2022building,
  author={Xu, Jingke and Wang, Libo},
  journal={IEEE Trans. Commun.}, 
  title={Building Capacity-Achieving {T-PIR} Schemes for Some Parameters Over Binary Field via Subfield Sub-Codes}, 
  year={2022},
  volume={70},
  number={1},
  pages={59-70},
  doi={10.1109/TCOMM.2021.3125056}}

@ARTICLE{tajeddine2019private,
  author={Tajeddine, Razane and Gnilke, Oliver W. and Karpuk, David and Freij-Hollanti, Ragnar and Hollanti, Camilla},
  journal={IEEE Transactions on Information Theory}, 
  title={Private Information Retrieval From Coded Storage Systems With Colluding, Byzantine, and Unresponsive Servers}, 
  year={2019},
  volume={65},
  number={6},
  pages={3898-3906},
  doi={10.1109/TIT.2018.2890285}}

@ARTICLE{Holzbaur2022towards,
  author={Holzbaur, Lukas and Freij-Hollanti, Ragnar and Li, Jie and Hollanti, Camilla},
  journal={IEEE Transactions on Information Theory}, 
  title={Toward the Capacity of Private Information Retrieval From Coded and Colluding Servers}, 
  year={2022},
  volume={68},
  number={1},
  pages={517-537},
  doi={10.1109/TIT.2021.3120316}}

\end{document}